\documentclass[11pt]{article}

\usepackage[margin=1in]{geometry}
\usepackage[pdftex]{graphicx}
\usepackage[update,prepend]{epstopdf}
\usepackage[justification=justified,singlelinecheck=false,font=small,labelfont=bf]{caption}
\usepackage{subcaption}
\usepackage{booktabs} %for ``\top(mid,bottom)rule''
\usepackage{tabularx}
\newcolumntype{Y}{>{\centering\arraybackslash}X} %center cells in tabularx environment
\usepackage{amsmath}
\allowdisplaybreaks[4]
\usepackage{amssymb}
\usepackage{amsfonts}
\usepackage{amsthm}
\usepackage{dsfont}
\usepackage{bold-extra}
\usepackage{titlesec}
\usepackage[title]{appendix}
\usepackage{cases}
\usepackage{enumerate}
\usepackage{verbatim}
\usepackage{xcolor}
\usepackage{mathtools}

\DeclarePairedDelimiter\floor{\lfloor}{\rfloor}
\allowdisplaybreaks[4]
\usepackage{hyperref}
\hypersetup{pdfpagemode=UseNone} % hide bookmarks tab automatically
\hypersetup{pdfstartview={XYZ null null 0.85}} % pdf zoom at 85%
\newcommand*{\MyLaw}{\mathrm{law}}
\newcommand*{\eqlawU}{\ensuremath{\mathop{\overset{\MyLaw}{=}}}} % unscaled version
\newcommand*{\eqlaw}{\mathop{\overset{\MyLaw}{\resizebox{\widthof{\eqlawU}}{\heightof{=}}{=}}}}

\newcommand\independent{\protect\mathpalette{\protect\independenT}{\perp}}
  \def\independenT#1#2{\mathrel{\rlap{$#1#2$}\mkern2mu{#1#2}}}

\newcommand{\RR}{\mathbb{R}}

\DeclareMathOperator*{\plim}{lim\vphantom{p}}

\DeclareMathOperator{\E}{\mathbb{E}}
\DeclareMathOperator{\Prob}{\mathbb{P}}
\DeclareMathOperator{\Ind}{\mathds{1}}

\newtheorem{theorem}{Theorem}[section]
\newtheorem{proposition}[theorem]{Proposition}

\newtheorem{remark}[theorem]{Remark}
\numberwithin{equation}{section}

\title{Convergence of an Euler scheme for a hybrid stochastic-local\\ volatility model with stochastic rates in foreign exchange markets}
\author{\scshape{andrei cozma}\thanks{\footnotesize\scshape{Mathematical Institute, University of Oxford, Oxford OX2 6GG, United Kingdom}\newline
\hspace*{1.8em}\{andrei.cozma, matthieu.mariapragassam, christoph.reisinger\}@maths.ox.ac.uk} \and \scshape{matthieu mariapragassam\footnotemark[1]} \and \scshape{christoph reisinger\footnotemark[1]}}

\date{}
\begin{document}
\maketitle

\begin{abstract}
\noindent
We study the Heston--Cox--Ingersoll--Ross\scalebox{.9}{\raisebox{.5pt}{++}} stochastic-local volatility model in the context of foreign exchange markets and propose a Monte Carlo simulation scheme which combines the full truncation Euler scheme for the stochastic volatility component and the stochastic domestic and foreign short interest rates with the log-Euler scheme for the exchange rate. We establish the exponential integrability of full truncation Euler approximations for the Cox--Ingersoll--Ross process and find a lower bound on the explosion time of these exponential moments. Under a full correlation structure and a realistic set of assumptions on the so-called leverage function, we prove the strong convergence of the exchange rate approximations and deduce the convergence of Monte Carlo estimators for a number of vanilla and path-dependent options. Then, we perform a series of numerical experiments for an autocallable barrier dual currency note.

\vspace{1em}

\vspace{.6em}

\vspace{.6em}
\end{abstract}

%%%%%%%%%%%%%%%%%%%%%%%%%%%%%%%%%%%%%%%%%%%%%%%%%%%%%%%%%%%%%%%%%%%%%%%%%%%%%%
%%% Section 1 %%%%%%%%%%%%%%%%%%%%%%%%%%%%%%%%%%%%%%%%%%%%%%%%%%%%%%%%%%%%%%%%%%%%%%
%%%%%%%%%%%%%%%%%%%%%%%%%%%%%%%%%%%%%%%%%%%%%%%%%%%%%%%%%%%%%%%%%%%%%%%%%%%%%%
\section{Introduction}\label{sec:intro}

The class of stochastic-local volatility (SLV) models have become very popular in the financial sector in recent years. They contain a stochastic volatility component as well as a local volatility component (called the leverage function) and combine advantages of the two. According to Ren et al. \cite{Ren:2007}, Tian et al. \cite{Tian:2015} and van der Stoep et al. \cite{Stoep:2014}, the general SLV model allows for a better calibration to European options and improves the pricing and risk-management performance when compared to pure local volatility (LV) or pure stochastic volatility (SV) models. We focus on the Heston SLV model because the Cox--Ingersoll--Ross (CIR) process for the variance is widely used in the industry due to its desirable properties, such as mean-reversion and non-negativity, and since semi-analytic formulae are available for calls and puts under Heston's model and can help calibrate the parameters easily. The local volatility component allows a perfect calibration to the market prices of vanilla options. At the same time, the stochastic volatility component already provides built-in smiles and skews which give a rough calibration, so that a relatively flat leverage function suffices for a perfect calibration.

In order to improve the pricing and hedging of foreign exchange (FX) options, we introduce stochastic domestic and foreign short interest rates into the model. Assuming constant interest rates is appealing due to its simplicity and does not lead to a serious mispricing of options with short maturities. However, empirical results \cite{Haastrecht:2009} have confirmed that the constant interest rate assumption is inappropriate for long-dated FX products, and the effect of interest rate volatility can be as relevant as that of the FX rate volatility for longer maturities. There has been a great deal of research carried out in the area of option pricing with stochastic volatility and interest rates in the past couple of years. For instance, Van Haastrecht et al. \cite{Haastrecht:2009} extended the model of Sch\"obel and Zhu \cite{Schobel:1999} to currency derivatives by including stochastic interest rates, a model that benefits from analytical tractability even in a full correlation setting due to the processes being Gaussian, whereas Ahlip and Rutkowski \cite{Ahlip:2013}, Grzelak and Oosterlee \cite{Grzelak:2011} and Van Haastrecht and Pelsser \cite{Haastrecht:2011} examined the Heston--2CIR/Vasicek hybrid models and concluded that a full correlation structure gives rise to a non-affine model even under a partial correlation of the driving Brownian motions.

There has also been an increasing interest in the calibration to vanilla options of models with stochastic and local volatility and stochastic interest rate dynamics. For instance, Guyon and Labord\'ere \cite{Guyon:2011} examined Monte Carlo-based calibration methods for a 3-factor SLV equity model with a stochastic domestic rate and discrete dividends, while Cozma, Mariapragassam and Reisinger \cite{Mariapragassam:2016} and Deelstra and Rayee \cite{Deelstra:2013} examined 4-factor hybrid SLV models with stochastic domestic and foreign rates and proposed different approaches to calibrate the leverage function. Hambly, Mariapragassam and Reisinger \cite{Mariapragassam:2015} took a different path and discussed the calibration of a 2-factor SLV model to barriers. As an aside, in this paper, the spot FX rate is defined as the number of units of domestic currency per unit of foreign currency.

The model of Cox et al. \cite{Cox:1985} is very popular when modeling short rates because the square-root (CIR) process admits a unique strong solution, is mean-reverting and analytically tractable. Cox et al. found the conditional distribution to be noncentral chi-squared and Broadie and Kaya \cite{Broadie:2006} proposed an efficient exact simulation scheme for the square-root process based on acceptance-rejection sampling. However, their algorithm presents a number of disadvantages such as complexity and lack of speed, and it is not fit to price strongly path-dependent options that require the value of the FX rate at a large number of time points. Furthermore, in the context of a stochastic-local volatility model, the correlations between the underlying processes make it difficult to simulate a noncentral chi-squared increment together with a correlated increment for the FX rate and the short rates, if applicable.

As of late, the non-negativity of the CIR process is considered to be less desirable when modeling short rates. On the one hand, the central banks have significantly reduced the interest rates since the 2008 financial crisis and it is now commonly accepted that interest rates need not be positive. On the other hand, if interest rates dropped too far below zero, then large amounts of money would be withdrawn from banks and government bonds, putting a severe squeeze on deposits. Hence, we model the domestic and foreign short rates using the shifted CIR (CIR\scalebox{.9}{\raisebox{.5pt}{++}}) process of Brigo and Mercurio \cite{Brigo:2001}. The CIR\scalebox{.9}{\raisebox{.5pt}{++}} model allows the short rates to become negative and can fit any observed term structure exactly while preserving the analytical tractability of the original model for bonds, caps, swaptions and other basic interest rate products.

In the present work, we put forward the Heston--Cox--Ingersoll--Ross\scalebox{.9}{\raisebox{.5pt}{++}} stochastic-local volatility (Heston--2CIR\scalebox{.9}{\raisebox{.5pt}{++}} SLV) model to price FX options. Independent of the correlation structure, the model is non-affine and hence a closed-form solution to the European option valuation problem is not available. Finite-difference methods are popular in finance and when the evolution of the exchange rate is governed by a complex system of stochastic differential equations (SDEs), it all comes down to solving a higher-dimensional partial differential equation (PDE). This can prove to be difficult due to the curse of dimensionality, because the number of grid points required increases exponentially with the number of dimensions. Monte Carlo algorithms are often preferred due to their ability to handle path-dependent features easily and there are numerous discretization schemes available, like the simple Euler--Maruyama scheme, see, e.g., Glasserman \cite{Glasserman:2003}. However, there are several disadvantages of this discretization, such as the fact that the approximation process can become negative with non-zero probability. In practice, one can set the process equal to zero when it turns negative -- called an absorption fix -- or reflect it in the origin -- referred to as a reflection fix. An overview of the Euler schemes considered thus far in the literature, including the full truncation scheme, can be found in Lord et al. \cite{Lord:2010}.

The usual theorems in Kloeden and Platen \cite{Kloeden:1999} on the convergence of numerical simulations require the drift and diffusion coefficients to be globally Lipschitz and satisfy a linear growth condition, whereas Higham et al. \cite{Higham:2002} extended the analysis to a simple Euler scheme for a locally Lipschitz SDE. The standard convergence theory does not apply to the CIR process since the square-root is not locally Lipschitz around zero. Consequently, alternative approaches have been employed to prove the weak or strong convergence of various discretizations for the square-root process. Deelstra and Delbaen \cite{Deelstra:1998}, Alfonsi \cite{Alfonsi:2005}, Higham and Mao \cite{Higham:2005} and Lord et al. \cite{Lord:2010} examined the strong global approximation and found either a logarithmic convergence rate or none at all. Strong convergence of order $1/2$ of the symmetrized and the backward (drift-implicit) Euler--Maruyama (BEM) schemes was established in Berkaoui et al. \cite{Berkaoui:2008} and Dereich et al. \cite{Dereich:2012}, respectively, albeit in a very restricted parameter regime for the symmetrized scheme. Alfonsi \cite{Alfonsi:2013} and Neuenkirch and Szpruch \cite{Szpruch:2014} recently showed that the BEM scheme for the SDE obtained through a Lamperti transformation is strongly convergent with rate one in the case of an inaccessible boundary point, while Hutzenthaler et al. \cite{Hutzenthaler:2014a} established a positive strong order of convergence in the case of an accessible boundary point.

Hutzenthaler et al. \cite{Hutzenthaler:2014b} identified a class of stopped increment-tamed Euler approximations for nonlinear systems of SDEs with locally Lipschitz drift and diffusion coefficients and proved that they preserve the exponential integrability of the exact solution under some mild assumptions, unlike the explicit, the linear-implicit or some tamed Euler schemes, which rarely do. However, the results of Hutzenthaler et al. do not apply to the present work beacuse the diffusion coefficient in the square-root model is not locally Lipschitz. In this work, we first prove that an explicit Euler scheme for the CIR process retains exponential integrability properties, which plays a key role in establishing the boundedness of moments of Euler approximations for SDE systems with CIR dynamics in one or more dimensions.

To the best of our knowledge, the convergence of Monte Carlo algorithms in a stochastic-local volatility context has yet to be established. Higham and Mao \cite{Higham:2005} considered an Euler simulation of the Heston model with a reflection fix in the diffusion coefficient to avoid negative values. They studied convergence properties of the stopped approximation process and used the boundedness of payoffs to prove weak convergence for a European put and an up-and-out call option. However, Higham and Mao mentioned that the arguments cannot be extended to cope with unbounded payoffs. We work under a different Euler scheme and overcome this problem by proving the uniform boundedness of moments of the true solution and its approximation, and then the strong convergence of the latter. The existence of moment bounds for Euler approximations of the Heston model is important for deriving strong convergence \cite{Higham:2002} and it has been a long-standing open problem until now. Furthermore, the existence of moment bounds for hybrid Heston-type stochastic-local volatility models plays an important role in the calibration routine \cite{Mariapragassam:2016}.

In this paper, we focus on the Heston stochastic-local volatility model with CIR\scalebox{.9}{\raisebox{.5pt}{++}} short interest rates and examine convergence properties of the Monte Carlo algorithm with the full truncation Euler (FTE) discretization for the squared volatility and the two short rates, and the log-Euler discretization for the exchange rate. We prefer the full truncation scheme introduced by Lord et al. \cite{Lord:2010} because it preserves the positivity of the original process, is easy to implement and is found empirically to produce the smallest bias of all explicit Euler schemes and to outperform the quasi-second order schemes of Kahl and J\"ackel \cite{Kahl:2006} and Ninomiya and Victoir \cite{Ninomiya:2008}. We also choose the FTE scheme over the BEM scheme since the latter works only in a restricted parameter regime that is often unrealistic, since the BEM scheme is equivalent to the implicit Milstein method employed by Kahl and J\"ackel to first order in the time step, and since the
quanto correction term in the dynamics of the foreign short rate would lead to technical challenges in the convergence analysis.

The quadratic exponential (QE) scheme of Andersen \cite{Andersen:2008} uses moment-matching techniques to approximate the noncentral chi-squared distributed square-root process and is an efficient alternative to the FTE scheme that typically produces a smaller bias at the expense of a more complex implementation. Empirical results suggest that the QE scheme outperforms the FTE scheme in the case of a severely violated Feller condition \cite{Haastrecht:2010}, whereas when the Feller condition is satisfied, the two schemes perform comparably well because the bias with the FTE scheme is small even when only a few number of time steps are used per year. Since there is no natural extension to the QE scheme to multi-dimensional Heston-type models that involve several correlated square-root processes such as the ones studied here, we prefer the FTE scheme in the subsequent convergence analysis.
%In our numerical experiments, we postulate a simple correlation structure of the model and observe that the calibrated Heston parameters violate the Feller condition. This motivates our choice in Section~\ref{sec:calibration} to employ the QE scheme to discretize the squared volatility and the FTE scheme to discretize the short rates.

In summary, to the best of our knowledge, this paper is the first to establish: (1) the exponential integrability of the FTE scheme for the CIR process; (2) the boundedness of moments of order greater than 1 of approximation schemes for Heston-type models; (3) the strong convergence of approximation schemes for models with stochastic-local volatility (with a bounded and Lipschitz leverage function) and stochastic CIR\scalebox{.9}{\raisebox{.5pt}{++}} short rates; (4) the weak convergence for options with unbounded payoffs, in particular, for European, Asian and barrier contracts (up to a critical time).
%In summary, we extend published convergence results for approximation schemes for the Heston model to derivatives with: (1) unbounded payoffs, for European and barrier contracts (up to a critical time); (2) stochastic-local volatility (with a bounded and Lipschitz leverage function); (3) stochastic CIR\scalebox{.9}{\raisebox{.5pt}{++}} short rates; (4) exotic payoffs (e.g. Asian options).

The remainder of this paper is structured as follows. In Section~\ref{sec:setup}, we introduce the model and discuss the postulated assumptions and an efficient calibration. In Section~\ref{sec:analysis}, we first define the simulation scheme and discuss the main theorem. Then, we investigate the uniform exponential integrability of the full truncation scheme for the square-root process and prove convergence of the exchange rate approximations. We conclude the section by establishing the convergence of Monte Carlo simulations for computing the expected discounted payoffs of European, Asian and barrier options. In Section~\ref{sec:numerics}, we carry out numerical experiments to justify our choice of model and to demonstrate convergence. Finally, Section~\ref{sec:conclusion} contains a short discussion.

%%%%%%%%%%%%%%%%%%%%%%%%%%%%%%%%%%%%%%%%%%%%%%%%%%%%%%%%%%%%%%%%%%%%%%%%%%%%%%
%%% Section 2 %%%%%%%%%%%%%%%%%%%%%%%%%%%%%%%%%%%%%%%%%%%%%%%%%%%%%%%%%%%%%%%%%%%%%%
%%%%%%%%%%%%%%%%%%%%%%%%%%%%%%%%%%%%%%%%%%%%%%%%%%%%%%%%%%%%%%%%%%%%%%%%%%%%%%
\section{Preliminaries}\label{sec:setup}

%%%%%%%%%%%%%%%%%%%%%%%%%%%%%%%%%%%%%%%%%%%%%%%%%%%%%%%%%%%%%%%%%%%%%%%%%%%%%%
\subsection{Model definition}\label{subsec:model}

In its most general form, we have in mind a model in an FX market, for the spot FX rate $S$, the squared volatility of the FX rate $v$, the domestic short interest rate $r^{d}$ and the foreign short interest rate $r^{f}$. Unless otherwise stated, in this paper, the subscripts and superscripts ``$d$'' and ``$f$'' indicate domestic and foreign, respectively. Consider a filtered probability space $\left(\Omega,\mathcal{F},\mathbb{P}\right)$ and suppose that the dynamics of the underlying processes are governed by the following system of SDEs under the domestic risk-neutral measure $\mathbb{Q}$:
\begin{align}\label{eq2.1}
	\begin{dcases}
	dS_{t} = \big(r^{d}_{t}-r^{f}_{t}\big)S_{t}dt + \sigma(t,S_t)\sqrt{v_{t}}\hspace{.5pt}S_{t}\hspace{1pt}dW^{s}_{t},\hspace{.75em} S_{0}>0, \\[2pt]
	dv_{t} \hspace{1pt} = k(\theta-v_{t})dt + \xi\sqrt{v_{t}}\,dW^{v}_{t},\hspace{.75em} v_{0}>0, \\[0.5pt]
	\hspace{5.5pt} r^{d}_{t} = g^{d}_{t} + h_{d}(t),  \\[-4pt]
	dg^{d}_{t} \hspace{-.5pt} = k_{d}(\theta_{d}-g^{d}_{t})dt + \xi_{d}\sqrt{g^{d}_{t}}\,dW^{d}_{t},\hspace{.75em} g^{d}_{0}>0, \\[-1.5pt]
	\hspace{5.5pt} r^{f}_{t} \hspace{-.5pt} = g^{f}_{t} + h_{f}(t),  \\[-4.5pt]
	dg^{f}_{t} \hspace{-1pt} = \big(k_{f}\theta_{f}-k_{f}g^{f}_{t}-\rho_{s\hspace{-.7pt}f}\xi_{f}\sigma(t,S_t)\sqrt{v_{t}g^{f}_{t}}\hspace{1pt}\big)dt + \xi_{f}\sqrt{g^{f}_{t}}dW^{f}_{t},\hspace{.75em} g^{f}_{0}>0,
	\end{dcases}
\end{align}
where $\{W^{s},W^{v},W^{d},W^{f}\}$ are standard Brownian motions, $\sigma$ is the leverage function and $h_{d}$ and $h_{f}$ are deterministic functions of time. The mean-reversion parameters $k$, $k_{d}$ and $k_{f}$, the long-term mean parameters $\theta$, $\theta_{d}$ and $\theta_{f}$, and the volatility parameters $\xi$, $\xi_{d}$ and $\xi_{f}$ are positive real numbers. The quanto correction term in the drift of the foreign short rate in \eqref{eq2.1} comes from changing from the foreign to the domestic risk-neutral measure \cite{Clark:2011}. As an aside, for Hull--White short rate processes, changing from the domestic spot measure to the domestic $T$-forward measure leads to a dimension reduction of the problem because the diffusion coefficients of the short rates do not depend on the level of the rates \cite{Grzelak:2012}. Since this is not the case for CIR short rate processes, we prefer to work under the domestic spot measure. Note that the above system can collapse to the Heston--2CIR\scalebox{.9}{\raisebox{.5pt}{++}} model if we set $\sigma=1$, or to a local volatility model with stochastic short rates if we set $k=\xi=0$. The standard Heston SLV model is the special case $k_{d}=\xi_{d}=k_{f}=\xi_{f}=h_{d,f}=0$. We can also think of \eqref{eq2.1} as a model in an equity market with stock price process $S$, stochastic interest rate $r^{d}$ and stochastic dividend yield $r^{f}$, in which case the quanto correction term vanishes. We consider a full correlation structure between the Brownian drivers $\{W^{s},W^{v},W^{d},W^{f}\}$, i.e., no assumptions on the constant correlation matrix $\Sigma$ are made, where
\begin{equation}\label{eq2.1.2}
\Sigma = \begin{bmatrix}
												1 & \rho_{sv} & \rho_{sd} & \rho_{s\hspace{-.7pt}f} \\
												\rho_{sv} & 1 & \rho_{v\hspace{.2pt}d} & \rho_{v\hspace{-.7pt}f} \\
												\rho_{sd} & \rho_{v\hspace{.2pt}d} & 1 & \rho_{d\hspace{.01pt}f} \\
												\rho_{s\hspace{-.7pt}f} & \rho_{v\hspace{-.7pt}f} & \rho_{d\hspace{.01pt}f} & 1
								\end{bmatrix}.
\end{equation}
Furthermore, we work under the following assumptions:

\vspace{.75em}
\noindent
{\rm ($\mathcal{A}$1)} The leverage function is bounded, i.e., there exists a non-negative constant $\sigma_{max}$ such that, for all $t \in [0,T]$ and $x \in [0,\infty)$, we have
\begin{equation}\label{eq2.2}
0 \leq \sigma(t,x) \leq \sigma_{max}\,.
\end{equation}
\noindent
{\rm ($\mathcal{A}$2)} There exist non-negative constants $A$, $B$ and a positive real number $\alpha$ such that, for all $t,u \in [0,T]$ and $x,y \in [0,\infty)$, we have
\begin{equation}\label{eq2.3}
\left|\sigma(t,x)-\sigma(u,y)\right| \leq A\left|t-u\right|^{\alpha} + B\left|x-y\right|.
\end{equation}
\noindent
{\rm ($\mathcal{A}$3)} There exists a non-negative constant $h_{max}$ such that, for all $t \in [0,T]$ and $i \in \left\{d,f\right\}$, we have
\begin{equation}\label{eq2.3.1}
\left|h_{i}(t)\right| \leq h_{max}\,.
\end{equation}

Hence, we assume that $\sigma$ is bounded, H\"older continuous in $t$ and Lipschitz in $S_t$. According to \cite{Mariapragassam:2016}, for the leverage function to be consistent with call and put prices, it has to be given by the formula \eqref{eq2.3.c}, which depends on the calibrated Dupire local volatility. In practice, the local volatility function usually arises as the interpolation of discrete values obtained from a discretized version of Dupire's formula. Hence, there is no loss of generality from a practical point of view in assuming that the leverage function is Lipschitz continuous and bounded on a compact subset of $\mathbb{R}_{+}^{2}$ of the form $[0,T]\times[x_{min},x_{max}]$, and furthermore that
\begin{equation}\label{eq2.4}
\sigma(t,x) = \sigma\big(t\wedge T,\hspace{1.5pt} x_{min}\Ind_{x\leq x_{min}} +\hspace{2pt} x\Ind_{x\in(x_{min},x_{max})} +\hspace{2pt} x_{max}\Ind_{x\geq x_{max}}\hspace{-2pt}\big).
\end{equation}
Then $\sigma$ is globally Lipschitz continuous and the second assumption holds with $\alpha=1$. We also assume that $h_{d,f}$ are bounded. According to \cite{Brigo:2001}, for a perfect fit to the initial term structure of interest rates, each of the two shift functions must be given by the difference between the (flat-forward) market instantaneous forward rate and a continuous function of time. Then $h_{d,f}$ are piecewise continuous and the third assumption holds.

%%%%%%%%%%%%%%%%%%%%%%%%%%%%%%%%%%%%%%%%%%%%%%%%%%%%%%%%%%%%%%%%%%%%%%%%%%%%%%
\subsection{Model calibration}\label{subsec:calibration}

The calibration of the 4-factor Heston--2CIR\scalebox{.9}{\raisebox{.5pt}{++}} SLV model \eqref{eq2.1} is of paramount importance and represents a mandatory step for the efficient pricing of derivative contracts. In \cite{Mariapragassam:2016}, we consider the 4-factor SLV model and propose a new calibration approach that builds on the particle method of \cite{Guyon:2011}, combined with a novel and efficient variance reduction technique that takes advantage of PDE calibration to increase its stability and accuracy. The numerical experiments in \cite{Mariapragassam:2016} suggest that this method almost recovers the calibration speed from the associated 2-factor SLV model with deterministic rates. We assume here a partial correlation structure where only $\rho_{sv}$, $\rho_{sd}$ and $\rho_{s\hspace{-.7pt}f}$ may be non-zero, and denote by $D^{d}$ and $D^{f}$ the domestic and foreign discount factors associated with their respective money market accounts, i.e., for any $t\in[0,T]$,
\begin{equation}\label{eq2.1.c}
D_{t}^{d}=e^{-\int_{0}^{t}{r^{d}_{u}}du} \hspace{3pt}\text{ and }\hspace{3pt} D_{t}^{f}=e^{-\int_{0}^{t}{r^{f}_{u}}du}.
\end{equation}
In addition to \eqref{eq2.1}, we consider a pure local volatility (LV) model,
\begin{equation}\label{eq2.2.c}
dS^{\text{LV}}_{t} = \big(f^{d}_{t}-f^{f}_{t}\big)S^{\text{LV}}_{t}dt + \sigma_{\text{LV}}\big(t,S^{\text{LV}}_t\big)S^{\text{LV}}_{t}\hspace{1pt}dW^{s}_{t},\hspace{.75em} S^{\text{LV}}_{0}=S_{0},
\end{equation}
where for $i\in\{d,f\}$, $f^{i}_{t} = -\hspace{1pt}\frac{\partial}{\partial t}\hspace{1pt}\log P^{i}(0,t)$ is the market instantaneous forward rate at time $0$ for a maturity $t$ and $P^{i}(0,t)$ is the market zero-coupon bond price at time $0$ for a maturity $t$.

Suppose that the LV model \eqref{eq2.2.c} has been calibrated and that $\sigma_{\text{LV}}$ has been determined. The main building block of the calibration routine is expressing the leverage function $\sigma$ in terms of the local volatility function $\sigma_{\text{LV}}$. We state the necessary and sufficient condition for a perfect calibration to vanilla options in the following theorem, proved in \cite{Mariapragassam:2016}.

\begin{theorem}[Theorem 1 in \cite{Mariapragassam:2016}]\label{Thm2.1}
If the spot process marginal density function under model \eqref{eq2.1} is continuous in space and under assumptions {\rm ($\mathcal{A}$1)} and {\rm ($\mathcal{A}$3)}, the call price under model \eqref{eq2.1} matches the market quote for any strike $K$ and maturity $T<T^{*}$ if and only if
\begin{align}\label{eq2.3.c}
\sigma^{2}(T,K) &=
\frac{\E\!\left[D_{T}^{d}\,|\,S_{T}=K\right]}{\E\!\left[D_{T}^{d}v_{T}\,|\,S_{T}=K\right]}
\Bigg\{\sigma_{\emph{\text{LV}}}^{2}(T,K) + \frac{2}{K^{2}\frac{\partial^{2}C_{\emph{\text{LV}}}}{\partial K^{2}}}\bigg(\E\!\Big[D_{T}^{d}\big(r^{f}_{T}-f^{f}_{T}\big)\big(S_{T}-K\big)^{+}\Big] \nonumber\\[2pt]
&-K\E\!\left[D_{T}^{d}\big(r^{d}_{T}-f^{d}_{T}\big)\Ind_{S_{T}\geq K}\right]+K\E\!\left[D_{T}^{d}\big(r^{f}_{T}-f^{f}_{T}\big)\Ind_{S_{T}\geq K}\right]\bigg)\Bigg\},
\end{align}
where $C_{\emph{\text{LV}}}$ is the call price under model \eqref{eq2.2.c}, all expectations are under the domestic risk-neutral measure, $\varphi=2+\sqrt{2}$, $\zeta=\xi\sigma_{max}$ and $T^{*}$ is as given below.
\begin{enumerate}[(1)]
\item{When $k<\varphi\zeta$,
\begin{equation}\label{eq2.4.c}
T^{*} = \frac{2}{\sqrt{\varphi^{2}\zeta^{2}-k^{2}}}\bigg[\frac{\pi}{2}+\arctan\bigg(\frac{k}{\sqrt{\varphi^{2}\zeta^{2}-k^{2}}}\bigg)\bigg].
\end{equation}}
\item{When $k\geq\varphi\zeta$,
\begin{equation}\label{eq2.5.c}
T^{*} = \infty\hspace{.5pt}.
\end{equation}}
\end{enumerate}
\end{theorem}

This result generalizes the formula in \cite{Guyon:2011} to a stochastic foreign short rate. Deelstra and Rayee \cite{Deelstra:2013} obtained a similar formula for a 4-factor SLV model with Hull--White short rate processes. An important step in the derivation is that the stochastic integral
\begin{equation*}
\int_{0}^{T}{\Ind_{S_{t}\geq K}\sigma(t,S_{t})\sqrt{v_{t}}\hspace{.5pt}D^{d}_{t}S_{t}\hspace{1pt}dW^{s}_{t}}
\end{equation*}
is a true martingale. On the one hand, $T^{*}$ is a lower bound on the explosion time of the second moment of the discounted spot process $D_{t}^{d}S_{t}$ (see Proposition \ref{Prop3.4.3}). On the other hand, moments of the Heston model can explode in finite time \cite{Andersen:2007}, a property that is inherited by the Heston-type model \eqref{eq2.1}. Note also Theorems \ref{Thm3.2.1} and \ref{Thm3.3.2} for possible moment explosions of the numerical approximations. Therefore, the formula \eqref{eq2.3.c} may not hold for some values of the model parameters and for large maturities $T$. However, in practice, $T^{*}$ is very large. For instance, following our calibration routine to EURUSD market data from \cite{Mariapragassam:2016}, we found $T^{*}=28.6$.%T^{*}=27.9

As an aside, note that if the volatility in model \eqref{eq2.1} is purely local, i.e., if $v=1$, we recover the formula in \cite{Clark:2011}, whereas if the two short rates are deterministic, i.e., if $g^{d}=g^{f}=0$, we recover the formula in \cite{Guyon:2011,Ren:2007}, i.e.,
\begin{equation}\label{eq2.6.c}
\sigma(T,K)=\frac{\sigma_{\text{LV}}(T,K)}{\sqrt{\E\!\left[v_{T}\,|\,S_{T}=K\right]}}\hspace{1pt}.
\end{equation}

In the computations in Section~\ref{sec:numerics}, we assume a simple correlation structure where only $\rho_{sv}$ may be non-zero. We use a five-step calibration routine which is explained in detail in \cite{Mariapragassam:2016}. First, we calibrate the LV model \eqref{eq2.2.c} and find $\sigma_{\text{LV}}$. We also calibrate the associated 2-factor SLV model with deterministic rates via PDE methods, using formula \eqref{eq2.6.c} and results from the literature \cite{Clark:2011,Ren:2007}. Independently, we calibrate the two CIR\scalebox{.9}{\raisebox{.5pt}{++}} processes under their respective markets. Then, we find the Heston parameters by calibrating the associated 4-factor SV model, i.e., the Heston--2CIR\scalebox{.9}{\raisebox{.5pt}{++}} SV model, using the closed-form solution for the call price in \cite{Ahlip:2013}, extended to CIR\scalebox{.9}{\raisebox{.5pt}{++}} processes, with other parameters from the above calibration. Finally, we calibrate $\sigma$ in the Heston--2CIR\scalebox{.9}{\raisebox{.5pt}{++}} SLV model \eqref{eq2.1} by working with the particle method of \cite{Guyon:2011} and extending it using the calibrated 2-factor SLV model with deterministic rates as a control variate. For a partial correlation structure with possibly non-zero $\rho_{sd}$ and $\rho_{s\hspace{-.7pt}f}$, we may use the historical values of the correlations between the spot FX rate and the short-term zero-coupon bonds, and extend the approximation formulae of \cite{Grzelak:2011} for the call price from a 3-factor Heston--CIR model to a 4-factor Heston--2CIR\scalebox{.9}{\raisebox{.5pt}{++}} model.

%In \cite{Mariapragassam:2016}, we used real world EURUSD market data and calibrated the 4-factor SLV model \eqref{eq2.1} to vanilla options up to 5 years with a high accuracy, with as few as 20\hspace{1pt}000 particles; we observed a maximum error in the absolute volatility of only $0.03\%$. We then compared the market-implied volatility curve with the model-implied volatility curves corresponding to the 4-factor SLV model and the 4-factor SV model, respectively. The importance of the local volatility component for a very good fit to vanilla options is illustrated in Figure \ref{fig:1}. Furthermore, the addition of stochastic rates to the model can have a significant impact on structured products, even more so when barrier features and coupon detachments are combined for longer-dated contracts. The reader can refer to \cite{Mariapragassam:2016} for more details, where we used the simulation scheme defined in the next section to price such products.

%%%%%%%%%%%%%%%%%%%%%%%%%%%%%%%%%%%%%%%%%%%%%%%%%%%%%%%%%%%%%%%%%%%%%%%%%%%%%%
%%% Section 3 %%%%%%%%%%%%%%%%%%%%%%%%%%%%%%%%%%%%%%%%%%%%%%%%%%%%%%%%%%%%%%%%%%%%%%
%%%%%%%%%%%%%%%%%%%%%%%%%%%%%%%%%%%%%%%%%%%%%%%%%%%%%%%%%%%%%%%%%%%%%%%%%%%%%%
\section{Convergence analysis}\label{sec:analysis}

%%%%%%%%%%%%%%%%%%%%%%%%%%%%%%%%%%%%%%%%%%%%%%%%%%%%%%%%%%%%%%%%%%%%%%%%%%%%%%
\subsection{The simulation scheme}\label{subsec:scheme}

We employ the full truncation Euler (FTE) scheme from \cite{Lord:2010} to discretize the variance and the two short rate processes. Consider the CIR process
\begin{equation}\label{eq2.5}
dy_{t} = k_{y}(\theta_{y}-y_{t})dt + \xi_{y}\sqrt{y_{t}}\,dW^{y}_{t}.
\end{equation}
Let $T$ be the maturity of the option under consideration and create an evenly spaced grid
\begin{equation*}
T = N\delta t,\;\, t_{n}=n\delta t,\hspace{1em} \forall\hspace{1pt} n\in\{0,1,...,N\}.
\end{equation*}
First, we introduce the discrete-time auxiliary process
\begin{equation}\label{eq2.6}
\tilde{y}_{t_{n+1}} = \tilde{y}_{t_{n}} + k_{y}(\theta_{y}-\tilde{y}_{t_{n}}^{+})\delta t + \xi_{y}\sqrt{\tilde{y}_{t_{n}}^{+}}\,\delta W^{y}_{t_{n}},
\end{equation}
where $y^{+} = \max\left(0,y\right)$ and $\delta W^{y}_{t_{n}} = W^{y}_{t_{n+1}} - W^{y}_{t_{n}}$, and its continuous-time interpolation
\begin{equation}\label{eq2.7}
\tilde{y}_{t} = \tilde{y}_{t_{n}} + k_{y}(\theta_{y}-\tilde{y}_{t_{n}}^{+})(t-t_{n}) + \xi_{y}\sqrt{\tilde{y}_{t_{n}}^{+}}\big(W^{y}_{t}-W^{y}_{t_{n}}\big),
\end{equation}
for any $t \in [t_{n},t_{n+1})$, as suggested in \cite{Higham:2005}. Then, we define the non-negative processes
\begin{numcases}{}
\hat{y}_{t} = \tilde{y}_{t}^{+} \label{eq2.8a} \\[2pt]
\bar{y}_{t} = \tilde{y}_{t_{n}}^{+} \label{eq2.8b}
\end{numcases}
whenever $t \in [t_{n},t_{n+1})$. Let $\bar{v}$ and $\bar{g}^{d}$ be the FTE discretizations of $v$ and $g^{d}$, respectively. Taking into account the presence of the quanto correction term in the drift of the foreign short rate, we similarly define
\begin{equation}\label{eq2.8c}
\tilde{g}_{t}^{f} = \tilde{g}_{t_{n}}^{f} + \Big[k_{f}\theta_{f}-k_{f}\big(\tilde{g}_{t_{n}}^{f}\big)^{+}-\rho_{s\hspace{-.7pt}f}\xi_{f}\sigma\big(t_{n},\bar{S}_{t_{n}}\big)\sqrt{\tilde{v}_{t_{n}}^{+}\big(\tilde{g}^{f}_{t_{n}}\big)^{+}}\hspace{1.5pt}\Big](t-t_{n}) + \xi_{f}\sqrt{\big(\tilde{g}_{t_{n}}^{f}\big)^{+}}\hspace{.5pt}\big(W^{f}_{t}-W^{f}_{t_{n}}\big),
\end{equation}
where $\bar{S}$ is the continuous-time approximation of $S$ defined below, as well as
\begin{numcases}{}
\hat{g}^{f}_{t} = \big(\tilde{g}_{t}^{f}\big)^{+} \label{eq2.8d}\\[2pt]
\bar{g}^{f}_{t} = \big(\tilde{g}_{t_{n}}^{f}\big)^{+} \label{eq2.8e}
\end{numcases}
whenever $t \in [t_{n},t_{n+1})$. For $i \in \left\{d,f\right\}$, the domestic and foreign short rate discretizations are
\begin{equation}\label{eq2.8f}
\bar{r}^{i}_{t} = \bar{g}^{i}_{t} + h_{i}(t).
\end{equation}
Finally, we use an Euler--Maruyama scheme to discretize the log-exchange rate. Let $x$ and $\bar{x}$ be the actual and the approximated log-processes, and let $\bar{S}=e^{\bar{x}}$ be the continuous-time approximation of $S$. Moreover, define $h = h_{d}-h_{f}$. Then the discrete method reads:
\begin{align}\label{eq2.8}
\bar{x}_{t_{n+1}} = \bar{x}_{t_{n}} + \int_{t_{n}}^{t_{n+1}}{\hspace{-2pt}h(u)du} + \Big(\bar{g}^{d}_{t_{n}} - \bar{g}^{f}_{t_{n}} - \frac{1}{2}\hspace{1pt}\sigma^{2}\big(t_{n},\bar{S}_{t_{n}}\big)\bar{v}_{t_{n}}\Big)\delta t + \sigma\big(t_{n},\bar{S}_{t_{n}}\big)\sqrt{\bar{v}_{t_{n}}}\,\delta W^{s}_{t_{n}}.
\end{align}
However, we find it convenient to work with the continuous-time approximation
\begin{align}\label{eq2.9}
\bar{x}_{t} = \bar{x}_{t_{n}} + \int_{t_{n}}^{t}{h(u)du} + \Big(\bar{g}^{d}_{t} - \bar{g}^{f}_{t} - \frac{1}{2}\hspace{1pt}\bar{\sigma}^{2}\big(t,\bar{S}_{t}\big)\bar{v}_{t}\Big)\big(t-t_{n}\big) + \bar{\sigma}\big(t,\bar{S}_{t}\big)\sqrt{\bar{v}_{t}}\,\Delta W^{s}_{t},
\end{align}
where $\Delta W^{s}_{t} = W^{s}_{t}-W^{s}_{t_{n}}$ and $\bar{\sigma}\big(t,\bar{S}_{t}\big)=\sigma\big(t_{n},\bar{S}_{t_{n}}\big)$ whenever $t\in[t_{n},t_{n+1})$. Hence,
\begin{align}\label{eq2.10}
\bar{x}_{t} = x_{0} + \int_{0}^{t}{\Big(\bar{r}^{d}_{u} - \bar{r}^{f}_{u} - \frac{1}{2}\hspace{1pt}\bar{\sigma}^{2}\big(u,\bar{S}_{u}\big)\bar{v}_{u}\Big) du} + \int_{0}^{t}{\bar{\sigma}\big(u,\bar{S}_{u}\big)\sqrt{\bar{v}_{u}}\,dW^{s}_{u}}.
\end{align}
Note that the convergence of the continuous-time approximation ensures that the discrete method approximates the true solution accurately at the gridpoints. Using It\^o's formula, we obtain
\begin{align}
\bar{S}_{t} = S_{0} + \int_{0}^{t}{\big(\bar{r}^{d}_{u} - \bar{r}^{f}_{u}\big)\bar{S}_{u}\hspace{1pt}du} + \int_{0}^{t}{\bar{\sigma}\big(u,\bar{S}_{u}\big)\sqrt{\bar{v}_{u}}\,\bar{S}_{u}\hspace{1.5pt}dW^{s}_{u}}. \label{eq2.12}
\end{align}
We prefer the log-Euler scheme to the standard Euler scheme to discretize the exchange rate process because the former preserves positivity and produces no discretization bias in the $S$-direction when $\sigma$ is constant. Furthermore, if $v$, $g^{d}$ and $g^{f}$ are also constant, then the log-Euler scheme is exact.

%%%%%%%%%%%%%%%%%%%%%%%%%%%%%%%%%%%%%%%%%%%%%%%%%%%%%%%%%%%%%%%%%%%%%%%%%%%%%%
\subsection{The main theorem}\label{subsec:theorem}

Define the arbitrage-free price of an option and its approximation under \eqref{eq2.12},
\begin{align}
U &= \E\Big[e^{-\int_{0}^{T}{r^{d}_{t} dt}}f(\hspace{.25pt}S\hspace{.75pt})\Big], \label{eq2.12a}\\[2pt]
\bar{U} &= \E\Big[e^{-\int_{0}^{T}{\bar{r}^{d}_{t} dt}}f(\hspace{.25pt}\bar{S}\hspace{.75pt})\Big]. \label{eq2.12b}
\end{align}
The payoff function $f$ may depend on the entire path of the underlying process and the expectation is under the domestic risk-neutral measure.

\begin{theorem}\label{Thm3.2.1} If $2k_{f}\theta_{f}>\xi_{f}^{2}$ and under assumptions {\rm ($\mathcal{A}$1)} to {\rm ($\mathcal{A}$3)}, the following statements hold.
\begin{enumerate}[(1)]
\item{The approximations to the values of the European put, the up-and-out barrier call and any barrier put option defined in \eqref{eq2.12b} converge as $\delta t \to 0$.}
\item{If $\zeta=\xi\sigma_{max}$ and
\begin{equation}\label{eq2.12c}
T < T^{*} \equiv \hspace{1pt}\frac{4k}{\zeta^{2}}\hspace{.5pt}\Ind_{\zeta<\hspace{.5pt}2k} \hspace{1pt}+\hspace{3pt} \frac{1}{\zeta-k}\hspace{.5pt}\Ind_{\zeta\geq\hspace{.5pt}2k}\hspace{1pt},
\end{equation}
then the approximations to the values of the European call, Asian options, the down-and-in/out and the up-and-in barrier call option defined in \eqref{eq2.12b} converge as $\delta t \to 0$.}
\end{enumerate}
\end{theorem}

\begin{remark}\label{Rmk3.2.2}
{\rm If assumption {\rm ($\mathcal{A}$1)} holds, then for the purpose of this paper we choose the smallest upper bound on the leverage function, namely 
\begin{equation}\label{eq2.13}
\sigma_{max}=\sup\big\{\sigma(t,x)\,|\hspace{2pt} t\in[0,T],\hspace{1pt} x\in[0,\infty)\big\}.
\end{equation}%
}
\end{remark}

\begin{remark}\label{Rmk3.2.3}
{\rm If the exchange rate and the foreign short rate dynamics are independent of each other, i.e., if $\rho_{s\hspace{-.7pt}f}=0$, then the quanto correction term in the drift of the foreign short rate vanishes and Theorem \ref{Thm3.2.1} holds independent of the Feller condition $2k_{f}\theta_{f}>\xi_{f}^{2}$, albeit with a simpler proof. If the domestic and foreign short rates are constant throughout the lifetime of the option and if, moreover, $\sigma(t,x)=1$ for all $t \in [0,T]$ and $x \in [0,\infty)$, then the system of equations \eqref{eq2.1} collapses to the Heston model and Theorem \ref{Thm3.2.1} applies with $\zeta=\xi$. This extends the convergence results of Higham and Mao \cite{Higham:2005} to options with unbounded payoffs.%
}
\end{remark}

\begin{remark}\label{Rmk3.2.4}
{\rm Tian et al. \cite{Tian:2015} calibrated the term-structure Heston SLV model parameters to the market implied volatility data on EURUSD using 10 term structures and maturities up to 5 years. The data in Table \ref{table2.1} suggest that the condition in \eqref{eq2.12c} is typically satisfied in practice, both for shorter and longer maturities.%
}
\begin{table}[htb]
\begin{center}
\caption{The calibrated Heston SLV parameters for EURUSD market data from August 23, 2012 \cite{Tian:2015}.}\label{table2.1}
\begin{tabularx}{\textwidth}{@{}YYYYYY@{}}
  \addlinespace[-5pt]
	\toprule[.1em]
  $T$ & $k$ & $\xi$ & $\sigma_{max}$ & $\zeta$ & $T^{*}$ \\
  \midrule
  $\hspace{.6em}1$ month &	$0.885$ &	$0.342$ &	$1.600$ & $0.547$ & $11.8$ years \\[2pt]
  $5$ years &	$0.978$ &	$0.499$ &	$1.300$ & $0.649$ &\hspace{1pt} $\hspace{1pt}9.3$ years \\
	\bottomrule[.1em]
	\addlinespace[3pt]
\end{tabularx}
\end{center}
\end{table}

{\rm In \cite{Mariapragassam:2016}, the Heston--2CIR\scalebox{.9}{\raisebox{.5pt}{++}} SLV model is calibrated to the EURUSD market data from March 18, 2016, for maturities up to $5$ years. The data in Table \ref{table2.2} indicate that the condition in \eqref{eq2.12c} is satisfied even for very long maturities. Furthermore, the following parameter values for the foreign short rate process are recovered: $k_{f}=0.011$, $\theta_{f}=1.166$, $\xi_{f}=0.037$, so that the Feller condition also holds, i.e., $2k_{f}\theta_{f}=0.0257\gg0.0014=\xi_{f}^{2}$.%
}
\begin{table}[htb]
\begin{center}
\caption{The calibrated Heston--2CIR++ SLV parameters for EURUSD market data from March 18, 2016 \cite{Mariapragassam:2016}.}\label{table2.2}
\begin{tabularx}{\textwidth}{@{}YYYYYY@{}}
  \addlinespace[-5pt]
	\toprule[.1em]
  $T$ & $k$ & $\xi$ & $\sigma_{max}$ & $\zeta$ & $T^{*}$ \\
  \midrule
  $5$ years & $1.412$ & $0.299$ & $1.399$ & $0.418$ & $32.3$ years \\
	\bottomrule[.1em]
	\addlinespace[3pt]
\end{tabularx}
\end{center}
\end{table}

{\rm In equity markets, the mean-reversion speed is usually several times greater than the volatility of volatility. For instance, Hurn et al. \cite{Hurn:2014} calibrated the Heston model for the S\&P 500 index from January 1990 to December 2011 using a combination of two out-of-the-money options, and found that $k=1.977\gg0.456=\xi$. Furthermore, we typically have $\zeta<1$, so that the condition in \eqref{eq2.12c} holds even for longer maturities.%
}
\end{remark}

%%%%%%%%%%%%%%%%%%%%%%%%%%%%%%%%%%%%%%%%%%%%%%%%%%%%%%%%%%%%%%%%%%%%%%%%%%%%%%
\subsection{The square-root process}\label{subsec:cir}

In order to prove the convergence of the approximation scheme in \eqref{eq2.12}, we first need to examine the stability of the moments of order higher than 1 of the actual and the discretized processes. However, this problem is directly related to the exponential integrability of the CIR process and its approximation. Let $y$ be the CIR process in \eqref{eq2.5} and $\bar{y}$ be the piecewise constant FTE interpolant as per \eqref{eq2.8b}.

\begin{proposition}\label{Prop3.3.1}
Let $\lambda>0$ and define the stochastic process
\begin{equation}\label{eq3.1}
\Theta_{t} \equiv \exp\bigg\{\lambda\int_{0}^{t}{y_{u}\,du}\bigg\},\hspace{1em} \forall\hspace{1pt}t\geq0.
\end{equation}
If $T<T^{*}$, then the first moment of $\Theta_{T}$ is bounded, i.e.,
\begin{equation}\label{eq3.2}
\E\big[\Theta_{T}\big] < \infty,
\end{equation}
where $T^{*}$ is as given below.
\begin{enumerate}[(1)]
\item{When $k_{y}<\sqrt{2\lambda}\hspace{1pt}\xi_{y}$,
\begin{equation}\label{eq3.3}
T^{*} = \frac{2}{\sqrt{\smash[b]{2\lambda\xi_{y}^{2}-k_{y}^{2}}}}\left[\frac{\pi}{2}+\arctan\bigg(\frac{k_{y}}{\sqrt{\smash[b]{2\lambda\xi_{y}^{2}-k_{y}^{2}}}}\bigg)\right].
\end{equation}}
\item{When $k_{y}\geq\sqrt{2\lambda}\hspace{1pt}\xi_{y}$,
\begin{equation}\label{eq3.4}
T^{*} = \infty\hspace{.5pt}.
\end{equation}}
\end{enumerate}
\end{proposition}
\begin{proof}
Follows directly from Proposition 3.1 in \cite{Andersen:2007}.
\qquad
\end{proof}

\begin{theorem}\label{Thm3.3.2}
Let $\lambda>0$ and define the stochastic process
\begin{equation}\label{eq3.7}
\bar{\Theta}_{t} \equiv \exp\bigg\{\lambda\int_{0}^{t}{\bar{y}_{u}\,du}\bigg\},\hspace{1em} \forall\hspace{1pt}t\geq0.
\end{equation}
If $T\leq T^{*}$ and $\delta_{T}<k_{y}^{-1}$, then the first moment of $\bar{\Theta}_{T}$ is uniformly bounded, i.e.,%
\begin{equation}\label{eq3.8}
\sup_{\delta t \in (0,\delta_{T})}\hspace{1.5pt}\E\big[\bar{\Theta}_{T}\big] < \infty,
\end{equation}
where $T^{*}$ is as given below.
\begin{enumerate}[(1)]
\item{When $k_{y}\leq\sqrt{0.5\lambda}\hspace{1pt}\xi_{y}$,
\begin{equation}\label{eq3.8.1}
T^{*} = \frac{1}{\sqrt{2\lambda}\hspace{1pt}\xi_{y}-k_{y}}\hspace{1pt}.
\end{equation}}
\item{When $k_{y}>\sqrt{0.5\lambda}\hspace{1pt}\xi_{y}$,
\begin{equation}\label{eq3.8.2}
T^{*} = \frac{2k_{y}}{\lambda\xi_{y}^{2}}\hspace{1pt}.
\end{equation}}
\end{enumerate}
\end{theorem}
\begin{proof}
First, we prove that there exists $\eta\geq1$ independent of $\delta t$ such that, for all $\omega\in[0,1]$,
\begin{equation}\label{eq3.0.1}
\eta^{2}\omega^{2}\lambda\xi_{y}^{2}T^{2} - 2\eta\omega k_{y}T - 2\eta + 2 \leq 0.
\end{equation}
Fix any $\eta\geq1$ and define the polynomial
\begin{equation}\label{eq3.0.2}
f_{\eta}(\omega) = \omega^{2}\eta^{2}\lambda\xi_{y}^{2}T^{2} - 2\omega\eta k_{y}T - 2(\eta-1),
\end{equation}
with two distinct real roots
\begin{equation}\label{eq3.0.3}
\omega_{\pm} = \frac{k_{y}\pm\sqrt{\smash[b]{k_{y}^{2}+2(\eta-1)\lambda\xi_{y}^{2}}}}{\eta\lambda\xi_{y}^{2}T}\hspace{1pt}.
\end{equation}
Since $\omega_{-}\leq0<\omega_{+}$, we have that $f_{\eta}([0,1])\leq0$ if and only if $f_{\eta}(1)\leq0$, i.e.,
\begin{equation}\label{eq3.0.4}
\eta^{2}\lambda\xi_{y}^{2}T^{2} - 2\eta(1+k_{y}T) + 2 \leq 0.
\end{equation}
This holds for some $\eta\geq1$ if and only if the quadratic polynomial in $\eta$ on the left-hand side has a real root greater or equal to one. Hence, we find the necessary and sufficient conditions:
\begin{equation*}
\big(\sqrt{2\lambda}\hspace{1pt}\xi_{y}-k_{y}\big)T \leq 1, \hspace{2pt}\text{ and }\hspace{4pt}
2\lambda\xi_{y}^{2}T \leq k_{y}+\sqrt{k_{y}^{2}+4\lambda\xi_{y}^{2}} \hspace{4pt}\text{ or }\hspace{4pt}
k_{y}+\sqrt{k_{y}^{2}+4\lambda\xi_{y}^{2}} < 2\lambda\xi_{y}^{2}T \leq 4k_{y},
\end{equation*}
which are equivalent to $T\leq T^{*}$, with $T^{*}$ defined in \eqref{eq3.8.1} -- \eqref{eq3.8.2}. Fix any $\eta$ satisfying \eqref{eq3.0.1}.

Next, we prove by induction on $0\hspace{-1pt}\leq\hspace{-1pt}m\hspace{-1pt}\leq\hspace{-1pt}N$ that, for all $\delta t<\delta_{T}$,
\begin{align}\label{eq3.9}
\E\big[\bar{\Theta}_{T}\big] &\leq
\exp\bigg\{\frac{\eta}{2}\big(k_{y}\theta_{y} + \nu_{y}\xi_{y}\big)\lambda(\delta t)^{2}(m-1)m\bigg\} \E\bigg[\exp\bigg\{\lambda\delta t \sum_{i=0}^{N-m-1}{\hspace{-.2em}\bar{y}_{t_{i}}} + \eta m\lambda\delta t\hspace{1pt}\bar{y}_{t_{N-m}}\bigg\}\bigg],
\end{align}
where
\begin{equation}\label{eq3.9.1}
\nu_{y} = \sqrt{\frac{\xi_{y}^{2}}{4\pi(1-k_{y}\delta_{T})^{2}}+\frac{1}{2\pi}\hspace{1pt}\sqrt{\frac{\xi_{y}^{4}}{4(1-k_{y}\delta_{T})^{4}} + \frac{k_{y}^{2}\theta_{y}^{2}}{(1-k_{y}\delta_{T})^{2}}}}\hspace{2pt}.
\end{equation}
Let $\left\{\mathcal{G}_{t}^{y},\hspace{1.5pt} 0\hspace{-1pt}\leq\hspace{-1pt}t\hspace{-1pt}\leq\hspace{-1pt}T\right\}$ be the natural filtration generated by the Brownian motion $W^{y}$ and consider the shorthand notation $\E_{t}^{y}\big[\hspace{.5pt}\cdot\hspace{.5pt}\big]=\E\big[\hspace{.5pt}\cdot\hspace{.5pt}|\mathcal{G}_{t}^{y}\big]$ for the conditional expectation. Note that \eqref{eq3.9} clearly holds when $m\in\{0,1\}$. Next, let us assume that \eqref{eq3.9} holds for $1\hspace{-1pt}\leq\hspace{-1pt}m\hspace{-1pt}<\hspace{-1pt}N$ and prove the inductive step. Conditioning on the $\sigma$-algebra $\mathcal{G}_{t_{N-m-1}}^{y}$, we obtain
\begin{align}
\E\big[\bar{\Theta}_{T}\big] &\leq
\exp\bigg\{\frac{1}{2}\hspace{1pt}\eta\big(k_{y}\theta_{y} + \nu_{y}\xi_{y}\big)\lambda(\delta t)^{2}(m-1)m\bigg\} \nonumber\\[0pt]
&\times \E\bigg[\exp\bigg\{\lambda\delta t\sum_{i=0}^{N-m-1}{\hspace{-.2em}\bar{y}_{t_{i}}}\bigg\}\E_{t_{N-m-1}}^{y}\hspace{-2pt}\bigg[\exp\bigg\{\eta m\lambda\delta t\hspace{1pt}\bar{y}_{t_{N-m}}\bigg\}\bigg]\bigg]. \label{eq3.10}
\end{align}
For convenience, define $w = \bar{y}_{t_{N-m-1}}$. If $Z\sim \mathcal{N}\left(0,1\right)$, then $\mathcal{G}_{t_{N-m-1}}^{y} \independent \delta W^{y}_{t_{N-m-1}} \eqlaw \sqrt{\delta t}\hspace{1pt} Z$. Let $\mathcal{I}$ be the inner expectation in \eqref{eq3.10}, then
\begin{align*}
\mathcal{I} \leq \E_{0,w}\left[\exp\left\{\eta m\lambda\delta t\hspace{1pt}\max\!\left[0,\,w + k_{y}(\theta_{y} - w)\delta t + \xi_{y}\sqrt{w\delta t}\hspace{1pt}Z\right]\right\}\right]. \label{eq3.11}
\end{align*}
There are two possible outcomes, namely $w=0$, in which case
\begin{equation}\label{eq3.11.1}
\mathcal{I} \leq \exp\Big\{\eta m k_{y}\theta_{y}\lambda(\delta t)^{2}\Big\},
\end{equation}
and $w>0$, which is treated now:
\begin{align}\label{eq3.11.2}
\mathcal{I} &\leq \int_{z_{0}}^{\infty}{\frac{1}{\sqrt{2\pi}}\exp\left\{-\frac{1}{2}z^{2} + \eta m\xi_{y}\lambda\sqrt{w}\hspace{1pt}(\delta t)^{3/2}z + \eta m\lambda\delta t \big[w + k_{y}(\theta_{y}-w)\delta t\big]\right\}dz} \nonumber\\[3pt]
&+ \int_{-\infty}^{z_{0}}{\frac{1}{\sqrt{2\pi}}\exp\left\{-\frac{1}{2}z^{2}\right\}dz},
\end{align}
where
\begin{equation}\label{eq3.11.3}
z_{0} = -\frac{k_{y}\theta_{y}\delta t + (1-k_{y}\delta t) w}{\xi_{y}\sqrt{w\delta t}}\hspace{1pt}.
\end{equation}
Recall that $\delta_{T}<k_{y}^{-1}$ and define
\begin{equation}\label{eq3.11.4}
z_{1} = z_{0} - \eta m\xi_{y}\lambda\sqrt{w}\hspace{1pt}(\delta t)^{3/2}.
\end{equation}
Then $z_{1} < z_{0} < 0$. If $\phi$ and $\Phi$ are the standard normal PDF and CDF, then
\begin{align*}
\mathcal{I} \leq \Phi\left(z_{0}\right) + \exp\bigg\{\eta m\lambda\delta t\big[k_{y}\theta_{y}\delta t + (1-k_{y}\delta t) w\big] + \frac{1}{2}\hspace{1pt}\eta^{2}m^{2}\xi_{y}^{2}\lambda^{2}w(\delta t)^{3}\bigg\}\Big\{1 - \Phi\left(z_{1}\right)\hspace{-2pt}\Big\},
\end{align*}
and hence
\begin{align}\label{eq3.11.5}
\mathcal{I} \leq \exp\bigg\{\eta m k_{y}\theta_{y}\lambda(\delta t)^{2} + \frac{1}{2}\hspace{1pt}a\lambda w\delta t\bigg\}\Big\{1 + \Phi\left(z_{0}\right) - \Phi\left(z_{1}\right)\hspace{-2pt}\Big\},
\end{align}
where
\begin{equation}\label{eq3.11.6}
a = 2\eta m(1-k_{y}\delta t) + \eta^{2}m^{2}\xi_{y}^{2}\lambda(\delta t)^{2}>0.
\end{equation}
Applying the mean value theorem to $\Phi\in C^{1}$, we can find $z \in \left[z_{1},z_{0}\right]$ such that
\begin{equation*}
\Phi(z_{0}) - \Phi(z_{1}) = \left(z_{0}-z_{1}\right)\phi\left(z\right) \leq \left(z_{0}-z_{1}\right)\phi\left(z_{0}\right).
\end{equation*}
Hence, from \eqref{eq3.11.3} and \eqref{eq3.11.4},
\begin{equation}\label{eq3.12}
\Phi(z_{0}) - \Phi(z_{1}) \leq \underbrace{\frac{1}{\sqrt{2\pi}}\hspace{1pt}\eta m\xi_{y}\lambda(\delta t)^{3/2}}_{= \, b\hspace{.5pt}, \text{ constant w.r.t. } w} \cdot\, \sqrt{w}\hspace{1pt}\exp\left\{-\frac{\big[k_{y}\theta_{y}\delta t + (1-k_{y}\delta t)w\big]^2}{2\hspace{1pt}\xi_{y}^{2}w\delta t}\right\}.
\end{equation}
We can think of the right-hand side as a function of $w$, say $f:(0,\infty)\mapsto\mathbb{R}$. Next, we show that
\begin{equation}\label{eq3.13}
f(w) \leq \eta m\nu_{y}\xi_{y}\lambda(\delta t)^{2},\hspace{1em} \forall\hspace{.5pt}w > 0.
\end{equation}
\begin{comment}
Note that the function $f$ is continuous and positive, and also that
\begin{equation*}
\plim_{w \to 0}f(w) = \plim_{w \to \infty}f(w) = 0.
\end{equation*}
\end{comment}
In order to find its global maximum, we need to compute the first derivative.
\begin{equation*}
f'(w) = b\hspace{.5pt}e^{-z_{0}^{2}/2} \left\{\frac{1}{2\sqrt{w}} + \sqrt{w}\left[\frac{k_{y}^{2}\theta_{y}^{2}\delta t}{2\hspace{1pt}\xi_{y}^{2}w^{2}} - \frac{(1 - k_{y}\delta t)^{2}}{2\hspace{1pt}\xi_{y}^2\delta t}\right]\right\}.
\end{equation*}
Therefore,
\begin{equation*}
f'(w) = 0 \hspace{5pt}\Leftrightarrow\hspace{5pt} -(1-k_{y}\delta t)^{2}w^{2} + \xi_{y}^{2}w\delta t + k_{y}^{2}\theta_{y}^{2}\left(\delta t\right)^{2} = 0.
\end{equation*}
To solve this quadratic, we divide throughout by $\left(\delta t\right)^{2}$ and introduce a new variable, $\hat{w} = w/\delta t$. Then there exists a unique, positive solution $\hat{w}_{0}$ for which the derivative is zero, namely
\begin{equation*}
\hat{w}_{0} = \frac{\xi_{y}^{2}}{2(1-k_{y}\delta t)^{2}}+\sqrt{\frac{\xi_{y}^{4}}{4(1-k_{y}\delta t)^{4}} + \frac{k_{y}^{2}\theta_{y}^{2}}{(1-k_{y}\delta t)^{2}}}\hspace{2pt}.
\end{equation*}
From \eqref{eq3.9.1}, we have that
\begin{equation*}
\hat{w}_{0} < 2\pi\nu_{y}^{2} \hspace{5pt}\Rightarrow\hspace{5pt} w_{0} < 2\pi\nu_{y}^{2}\delta t.
\end{equation*}
Since the second root is negative, $f(w)$ must be increasing up to $w_{0}$ and decreasing after this point, so the function attains its global maximum at $w_{0}$. Hence,
\begin{equation*}
f(w) \leq f(w_{0}) = b\sqrt{w_{0}}\hspace{1pt}e^{-z_{0}^{2}/2} \leq b\sqrt{2\pi\nu_{y}^{2}\delta t} = \eta m\nu_{y}\xi_{y}\lambda(\delta t)^{2}.
\end{equation*}
Making use of the upper bound in \eqref{eq3.13}, we derive the following inequality:
\begin{equation}\label{eq3.16}
1 + \Phi(z_{0}) - \Phi(z_{1}) \leq \exp\Big\{\eta m\nu_{y}\xi_{y}\lambda(\delta t)^{2}\Big\}.
\end{equation}
Substituting back into \eqref{eq3.11.5} with \eqref{eq3.16}, we get
\begin{equation*}
\mathcal{I} \leq \exp\bigg\{\eta m\big(k_{y}\theta_{y}+\nu_{y}\xi_{y}\big)\lambda(\delta t)^{2} + \frac{1}{2}\hspace{1pt}a\lambda w\delta t\bigg\}.
\end{equation*}
Note from \eqref{eq3.11.1} that this holds when $w=0$ as well. Applying \eqref{eq3.0.1} with $\omega=\frac{m}{N}$ leads to
\begin{equation*}
\eta^{2}m^{2}\xi_{y}^{2}\lambda(\delta t)^{2} - 2\eta m k_{y}\delta t - 2\eta + 2 \leq 0.
\end{equation*}
Hence, from \eqref{eq3.11.6},
\begin{equation*}
a \leq 2\eta(m+1) - 2.
\end{equation*}
Therefore,
\begin{equation*}
\mathcal{I} \leq \exp\bigg\{\eta m\big(k_{y}\theta_{y}+\nu_{y}\xi_{y}\big)\lambda(\delta t)^{2} - \lambda w\delta t + \eta(m+1)\lambda w\delta t\bigg\}.
\end{equation*}
Substituting back into \eqref{eq3.10} gives the inductive step. Finally, taking $m=N$ in \eqref{eq3.9} leads to
\begin{equation}\label{eq3.17.2}
\E\big[\bar{\Theta}_{T}\big] < \exp\left\{\frac{1}{2}\hspace{1pt}\eta\lambda T^{2}\big(k_{y}\theta_{y}+\nu_{y}\xi_{y}\big) + \eta\lambda T y_{0}\right\}.
\end{equation}
The right-hand side is finite and independent of $\delta t$ and the conclusion follows. %\hfill
\qquad
\end{proof}

The next two results establish uniform moment bounds for both the original and the discretized variance and short rate processes. Note that under assumption {\rm ($\mathcal{A}$3)}, the moment boundedness of $g^{i}$ and $\hat{g}^{i}$, for $i\in\{d,f\}$, extends naturally to $|r^{i}|$ and $|\bar{r}^{i}|$.

\begin{proposition}\label{Prop3.3.3}
The following two statements hold.%\hspace{.5pt} % avoid latex warning
\begin{enumerate}[(1)]
\item{The square-root process in \eqref{eq2.5} has uniformly bounded moments, i.e.,
\begin{equation}\label{eq3.18.1}
\E\bigg[\sup_{t\in[0,T]} y_{t}^{p}\bigg] < \infty,\hspace{1em} \forall\hspace{.5pt}p\geq1.
\end{equation}}
\item{The FTE scheme in \eqref{eq2.8a} for the square-root process has uniformly bounded moments, i.e.,
\begin{equation}\label{eq3.18.2}
\sup_{\delta t \in (0,\delta_{T})}\E\bigg[\sup_{t \in [0,T]}\hat{y}_{t}^p\bigg] < \infty,\hspace{1em} \forall\hspace{.5pt}p\geq1,\hspace{.5em} \forall\hspace{.5pt}\delta_{T}>0.
\end{equation}}
\end{enumerate}
\end{proposition}
\begin{proof}
(1) The polynomial moments of the square-root process can be expressed in terms of the confluent hypergeometric function and, according to Theorem 3.1 in \cite{Hurd:2008} or to \cite{Dereich:2012}, they are uniformly bounded, i.e.,
\begin{equation}\label{eq3.19.1}
\sup_{t \in [0,T]}\E\big[y_{t}^p\big] < \infty,\hspace{1em} \forall\hspace{.5pt}p>-\hspace{1pt}\frac{2k_{y}\theta_{y}}{\xi_{y}^{2}}\hspace{1pt}.
\end{equation}
On the other hand, using H\"older's inequality, we deduce that
\begin{equation}\label{eq3.19.2}
\sup_{t\in[0,T]} y_{t}^{p} \leq 3^{p-1}\big(y_{0} + k_{y}\theta_{y}T\big)^{p} + 3^{p-1}k_{y}^{p}T^{p-1}\int_{0}^{T}{y_{u}^{p}\,du} + 3^{p-1}\xi_{y}^{p}\sup_{t\in[0,T]}\bigg|\int_{0}^{t}{\sqrt{y_{u}}\,dW^{y}_{u}}\bigg|^{p}.
\end{equation}
Using the Burkholder--Davis--Gundy (BDG) inequality, we can find a constant $C_{p}>0$ such that
\begin{equation}\label{eq3.19.3}
\E\bigg[\sup_{t\in[0,T]}\bigg|\int_{0}^{t}{\sqrt{y_{u}}\,dW^{y}_{u}}\bigg|^{p}\hspace{1pt}\bigg] \leq C_{p}\E\bigg[\bigg(\int_{0}^{T}{y_{u}\,du}\bigg)^{p/2}\hspace{1pt}\bigg] \leq \frac{1}{2}\hspace{1pt}C_{p} + \frac{1}{2}\hspace{1pt}C_{p}T^{p-1}\E\bigg[\int_{0}^{T}{y_{u}^{p}\,du}\bigg].
\end{equation}
Taking expectations in \eqref{eq3.19.2} and employing Fubini's theorem, we get
\begin{equation}\label{eq3.19.4}
\E\bigg[\sup_{t\in[0,T]} y_{t}^{p}\bigg] \leq 3^{p-1}\big(y_{0} + k_{y}\theta_{y}T\big)^{p} + \frac{1}{2}\hspace{.5pt}3^{p-1}\xi_{y}^{p}C_{p} + 3^{p-1}\big(k_{y}^{p}+0.5\hspace{.5pt}\xi_{y}^{p}C_{p}\big)T^{p}\sup_{t \in [0,T]}\E\big[y_{t}^p\big].
\end{equation}
The right-hand side is finite by \eqref{eq3.19.1}, whence the conclusion.

(2) Integrating the auxiliary process $\tilde{y}$ defined in \eqref{eq2.7}, we deduce that
\begin{equation}\label{eq3.20}
\tilde{y}_{t} = y_{0} + k_{y}\int_{0}^{t}{\left(\theta_{y}-\bar{y}_{u}\right)du} + \xi_{y}\int_{0}^{t}{\sqrt{\bar{y}_{u}}\,dW^{y}_{u}}.
\end{equation}
For any $p,\epsilon>0$, there exists a constant $c(p,\epsilon)>0$ such that $\max(0,x)^{p}\leq c(p,\epsilon)e^{\epsilon x}$ for all $x\in \mathbb{R}$. In particular, this implies that $\hat{y}_{t}^{p}\leq c(p,\epsilon)e^{\epsilon \tilde{y}_{t}}$ for all $t\in[0,T]$. Hence,
\begin{equation}\label{eq3.20.1}
\sup_{t \in [0,T]}\hspace{-1pt}\E\big[\hat{y}_{t}^p\big] \leq c(p,\epsilon)\hspace{1pt}e^{\epsilon\left(y_{0}+k_{y}\theta_{y}T\right)}\hspace{-1pt}\sup_{t \in [0,T]}\hspace{-1pt}\E\big[\Gamma_{t}\big],
\end{equation}
where
\begin{equation}\label{eq3.20.2}
\Gamma_{t} \equiv \exp\bigg\{\hspace{-1.5pt}-\epsilon k_{y}\hspace{-1pt}\int_{0}^{t}{\bar{y}_{u}\,du} + \epsilon\hspace{.5pt}\xi_{y}\hspace{-1pt}\int_{0}^{t}{\hspace{-1.5pt}\sqrt{\bar{y}_{u}}\,dW^{y}_{u}}\bigg\}.
\end{equation}
Assuming that $t \in [t_{n}, t_{n+1}]$ and conditioning on the $\sigma$-algebra $\mathcal{G}_{t_{n}}^{y}$, we get
\begin{equation}\label{eq3.20.3}
\E_{t_{n}}^{y}\!\big[\Gamma_{t}\big] = \Gamma_{t_{n}}\hspace{.5pt}\exp\Big\{\big(0.5\epsilon^{2}\xi_{y}^{2}-\epsilon k_{y}\big)\big(t-t_{n}\big)\bar{y}_{t_{n}}\Big\}.
\end{equation}
However, we can find $\epsilon$ sufficiently small such that $k_{y}\geq0.5\epsilon\hspace{.5pt}\xi_{y}^{2}$. Then $\E_{t_{n}}^{y}\!\big[\Gamma_{t}\big] \leq \E_{t_{n}}^{y}\!\big[\Gamma_{t_{n}}\big]$ and the law of iterated expectations ensures that
\begin{equation*}
\E\big[\Gamma_{t}\big] \leq \E\big[\Gamma_{t_{n}}\big] \leq \E\big[\Gamma_{t_{n-1}}\big] \leq \hdots \leq \E\big[\Gamma_{0}\big] \;\,\Rightarrow\; \sup_{t \in [0,T]}\E\big[\Gamma_{t}\big] = 1.
\end{equation*}
Substituting back into \eqref{eq3.20.1}, we deduce that
\begin{equation}\label{eq3.20.4}
\sup_{t \in [0,T]}\hspace{-1pt}\E\big[\hat{y}_{t}^p\big] \leq c(p,\epsilon)\hspace{1pt}e^{\epsilon\left(y_{0}+k_{y}\theta_{y}T\right)},
\end{equation}
which is independent of $\delta t$. On the other hand, since $\hat{y}_{t}\leq|\tilde{y}_{t}|$ and using H\"older's inequality, we get
\begin{equation}\label{eq3.20.5}
\sup_{t \in [0,T]}\hat{y}_{t}^p \leq 3^{p-1}\big(y_{0} + k_{y}\theta_{y}T\big)^{p} + 3^{p-1}k_{y}^{p}T^{p-1}\int_{0}^{T}{\bar{y}_{u}^{p}\hspace{1.5pt}du} + 3^{p-1}\xi_{y}^{p}\sup_{t\in[0,T]}\bigg|\int_{0}^{t}{\sqrt{\bar{y}_{u}}\,dW^{y}_{u}}\bigg|^{p}.
\end{equation}
Taking expectations on both sides and using the BDG inequality as in \eqref{eq3.19.3} leads to
\begin{equation}\label{eq3.20.6}
\E\bigg[\sup_{t \in [0,T]}\hat{y}_{t}^p\bigg] \leq 3^{p-1}\big(y_{0} + k_{y}\theta_{y}T\big)^{p} + \frac{1}{2}\hspace{.5pt}3^{p-1}\xi_{y}^{p}C_{p} + 3^{p-1}\big(k_{y}^{p}+0.5\hspace{.5pt}\xi_{y}^{p}C_{p}\big)T^{p}\sup_{t \in [0,T]}\E\big[\hat{y}_{t}^p\big],
\end{equation}
for some constant $C_{p}>0$. However, $\sup_{t \in [0,T]}\E\big[\bar{y}_{t}^p\big]\leq\sup_{t \in [0,T]}\E\big[\hat{y}_{t}^p\big]$, and the conclusion follows from \eqref{eq3.20.4}.
\qquad
\end{proof}

\begin{proposition}\label{Prop3.3.4}
Under assumption {\rm ($\mathcal{A}$1)}, the following two statements hold.
\begin{enumerate}[(1)]
\item{The process $g^{f}$ has uniformly bounded moments, i.e.,
\begin{equation}\label{eq3.21.1}
\E\bigg[\sup_{t\in[0,T]}\big(g^{f}_{t}\big)^{p}\bigg] < \infty,\hspace{1em} \forall\hspace{.5pt}p\geq1.
\end{equation}}
\item{The process $\hat{g}^{f}$ from \eqref{eq2.8d} has uniformly bounded moments, i.e.,
\begin{equation}\label{eq3.21.2}
\sup_{\delta t\in(0,\delta_{T})}\E\bigg[\sup_{t\in[0,T]}\big(\hat{g}^{f}_{t}\big)^p\bigg] < \infty,\hspace{1em} \forall\hspace{.5pt}p\geq1,\hspace{.5em} \forall\hspace{.5pt}\delta_{T}>0.
\end{equation}}
\end{enumerate}
\end{proposition}
\begin{proof}
(1) From \eqref{eq2.1}, we know that
\begin{equation}\label{eq3.22.1}
g^{f}_{t} = g^{f}_{0} + k_{f}\theta_{f}t - k_{f}\int_{0}^{t}{g^{f}_{u}\,du} - \rho_{s\hspace{-.7pt}f}\xi_{f}\int_{0}^{t}{\sigma(u,S_{u})\sqrt{v_{u}g^{f}_{u}}\,du} + \xi_{f}\int_{0}^{t}{\sqrt{g^{f}_{u}}\,dW^{f}_{u}}.
\end{equation}
Since $g^{f}$ is non-negative and using the fact that $2\sqrt{|ab|}\leq |a|+|b|$, we find an upper bound
\begin{equation*}
g^{f}_{t} \leq g^{f}_{0} + k_{f}\theta_{f}t + \frac{1}{2}\hspace{1pt}|\rho_{s\hspace{-.7pt}f}|\xi_{f}\sigma_{max}\int_{0}^{t}{v_{u}\,du} + \frac{1}{2}\big(2k_{f}+|\rho_{s\hspace{-.7pt}f}|\xi_{f}\sigma_{max}\big)\int_{0}^{t}{g^{f}_{u}\,du} + \xi_{f}\bigg|\int_{0}^{t}{\sqrt{g^{f}_{u}}\,dW^{f}_{u}}\bigg|.
\end{equation*}
Fix $t\in[0,T]$. Using H\"older's inequality, we get
\begin{align}\label{eq3.22.2}
\sup_{s\in[0,t]}\big(g^{f}_{s}\big)^{p} &\leq 4^{p-1}\big(g^{f}_{0} + k_{f}\theta_{f}T\big)^{p} + 2^{p-2}|\rho_{s\hspace{-.7pt}f}|^{p}\xi_{f}^{p}\sigma_{max}^{p}T^{p-1}\int_{0}^{T}{v_{u}^{p}\,du} \nonumber\\[0pt]
&\hspace{-1.5em}+ 2^{p-2}\big(2k_{f}+|\rho_{s\hspace{-.7pt}f}|\xi_{f}\sigma_{max}\big)^{p}T^{p-1}\int_{0}^{t}{\big(g^{f}_{u}\big)^{p}du} + 4^{p-1}\xi_{f}^{p}\sup_{s\in[0,t]}\bigg|\int_{0}^{s}{\sqrt{g^{f}_{u}}\,dW^{f}_{u}}\bigg|^{p}.
\end{align}
Taking expectations on both sides and using the BDG inequality as in \eqref{eq3.19.3} leads to
\begin{align}\label{eq3.22.3}
\E\bigg[\sup_{s\in[0,t]}\big(g^{f}_{s}\big)^{p}\bigg] &\leq 4^{p-1}\big(g^{f}_{0} + k_{f}\theta_{f}T\big)^{p} + 2^{2p-3}\xi_{f}^{p}C_{p} + 2^{p-2}|\rho_{s\hspace{-.7pt}f}|^{p}\xi_{f}^{p}\sigma_{max}^{p}T^{p}\sup_{u\in[0,T]}\E\big[v_{u}^{p}\big] \nonumber\\[0pt]
&\hspace{-1.5em}+ \Big(2^{p-2}\big(2k_{f}+|\rho_{s\hspace{-.7pt}f}|\xi_{f}\sigma_{max}\big)^{p}T^{p-1} + 2^{2p-3}\xi_{f}^{p}C_{p}T^{p-1}\Big)\int_{0}^{t}{\E\bigg[\sup_{s\in[0,u]}\big(g^{f}_{s}\big)^{p}\bigg]du},
\end{align}
for some constant $C_{p}>0$. Applying Gronwall's inequality (see \cite{Dragomir:2003}), we get
\begin{align}\label{eq3.22.4}
\E\bigg[\sup_{t\in[0,T]}\big(g^{f}_{t}\big)^{p}\bigg] &\leq \Big(4^{p-1}\big(g^{f}_{0} + k_{f}\theta_{f}T\big)^{p} + 2^{2p-3}\xi_{f}^{p}C_{p} + 2^{p-2}|\rho_{s\hspace{-.7pt}f}|^{p}\xi_{f}^{p}\sigma_{max}^{p}T^{p}\sup_{u\in[0,T]}\E\big[v_{u}^{p}\big]\Big) \nonumber\\[0pt]
&\times \exp\left\{2^{p-2}\big(2k_{f}+|\rho_{s\hspace{-.7pt}f}|\xi_{f}\sigma_{max}\big)^{p}T^{p} + 2^{2p-3}\xi_{f}^{p}C_{p}T^{p}\right\}.
\end{align}
The conclusion follows from Proposition \ref{Prop3.3.3}.

(2)
Integrating the auxiliary process $\tilde{g}^{f}$ defined in \eqref{eq2.8c}, we deduce that
\begin{equation}\label{eq3.23.1}
\tilde{g}^{f}_{t} = g^{f}_{0} + k_{f}\theta_{f}t - k_{f}\int_{0}^{t}{\bar{g}^{f}_{u}\,du} - \rho_{s\hspace{-.7pt}f}\xi_{f}\int_{0}^{t}{\bar{\sigma}\big(u,\bar{S}_{u}\big)\sqrt{\bar{v}_{u}\bar{g}^{f}_{u}}\,du} + \xi_{f}\int_{0}^{t}{\sqrt{\bar{g}^{f}_{u}}\,dW^{f}_{u}}.
\end{equation}
Since $\hat{g}^{f}_{t}\leq|\tilde{g}^{f}_{t}|$, following the previous argument, we deduce that
\begin{align}\label{eq3.23.2}
\E\bigg[\sup_{s\in[0,t]}\big(\hat{g}^{f}_{s}\big)^{p}\bigg] &\leq 4^{p-1}\big(g^{f}_{0} + k_{f}\theta_{f}T\big)^{p} + 2^{2p-3}\xi_{f}^{p}C_{p} + 2^{p-2}|\rho_{s\hspace{-.7pt}f}|^{p}\xi_{f}^{p}\sigma_{max}^{p}T^{p}\sup_{u\in[0,T]}\E\big[\bar{v}_{u}^{p}\big] \nonumber\\[0pt]
&\hspace{-1em}+ \Big(2^{p-2}\big(2k_{f}+|\rho_{s\hspace{-.7pt}f}|\xi_{f}\sigma_{max}\big)^{p}T^{p-1} + 2^{2p-3}\xi_{f}^{p}C_{p}T^{p-1}\Big)\int_{0}^{t}{\E\Big[\big(\bar{g}^{f}_{u}\big)^{p}\Big]du}.
\end{align}
Since $\sup_{u\in[0,T]}\E\big[\bar{v}_{u}^{p}\big]\leq\sup_{u\in[0,T]}\E\big[\hat{v}_{u}^{p}\big]$ and $\bar{g}^{f}_{u}\leq\sup_{s\in[0,u]}\hat{g}^{f}_{s}$, using Gronwall's inequality, we get
\begin{align}\label{eq3.23.3}
\E\bigg[\sup_{t\in[0,T]}\big(\hat{g}^{f}_{t}\big)^{p}\bigg] &\leq \Big(4^{p-1}\big(g^{f}_{0} + k_{f}\theta_{f}T\big)^{p} + 2^{2p-3}\xi_{f}^{p}C_{p} + 2^{p-2}|\rho_{s\hspace{-.7pt}f}|^{p}\xi_{f}^{p}\sigma_{max}^{p}T^{p}\sup_{u\in[0,T]}\E\big[\hat{v}_{u}^{p}\big]\Big) \nonumber\\[0pt]
&\times \exp\left\{2^{p-2}\big(2k_{f}+|\rho_{s\hspace{-.7pt}f}|\xi_{f}\sigma_{max}\big)^{p}T^{p} + 2^{2p-3}\xi_{f}^{p}C_{p}T^{p}\right\}.
\end{align}
The conclusion follows from Proposition \ref{Prop3.3.3}.
\qquad
\end{proof}

Next, we establish the strong mean square convergence of the discretized variance and domestic short rate processes. Unlike in \cite{Lord:2010}, where the focus was on the continuous-time approximation $\hat{y}$, we are rather interested in the behaviour of $\bar{y}$ in the limit as the time step goes to zero. Note that the convergence of $\bar{g}^{d}$ extends naturally to $\bar{r}^{d}$.

\begin{proposition}\label{Prop3.3.5} The full truncation scheme in \eqref{eq2.8b} for the square-root process converges strongly in the $L^{2}$ sense, i.e.,
\begin{equation}\label{eq3.24}
\plim_{\delta t \to 0} \sup_{t \in \left[0,T\right]} \E \big[|y_{t}-\bar{y}_{t}|^{2}\big] = 0.
\end{equation}
\end{proposition}
\begin{proof}
First, fix $t\in[0,T]$. Since $|y_{t}-\hat{y}_{t}| \leq |y_{t}-\tilde{y}_{t}|$, using Cauchy's inequality, we get
\begin{equation}\label{eq3.25.1}
\sup_{s\in\left[0,t\right]}|y_{s}-\hat{y}_{s}|^{2} \leq 2k_{y}^{2}t\int_{0}^{t}{(y_{u}-\bar{y}_{u})^{2}du} + 2\xi_{y}^{2}\sup_{s\in\left[0,t\right]}\bigg|\int_{0}^{s}{\big(\sqrt{y_{u}}-\sqrt{\bar{y}_{u}}\hspace{1pt}\big)dW_{u}^{y}}\bigg|^{2}.
\end{equation}
Taking expectations and then using Doob's martingale inequality and Fubini's theorem leads to
\begin{align}\label{eq3.25.2}
\E\bigg[\sup_{s\in\left[0,t\right]}|y_{s}-\hat{y}_{s}|^{2}\bigg] &\leq 4k_{y}^{2}T\int_{0}^{t}{\E\bigg[\sup_{s\in\left[0,u\right]}|y_{s}-\hat{y}_{s}|^{2}\bigg]du} + 4k_{y}^{2}T^{2}\sup_{u\in\left[0,T\right]}\E\big[|\hat{y}_{u}-\bar{y}_{u}|^{2}\big] \nonumber\\[3pt]
&+ 8\xi_{y}^{2}T\sup_{u\in\left[0,T\right]}\E\big[|y_{u}-\hat{y}_{u}|\big] + 8\xi_{y}^{2}T\sup_{u\in\left[0,T\right]}\E\big[|\hat{y}_{u}-\bar{y}_{u}|\big].
\end{align}
Applying Gronwall's inequality, we get
\begin{align}\label{eq3.25.3}
\E\bigg[\sup_{t\in\left[0,T\right]}|y_{t}-\hat{y}_{t}|^{2}\bigg] &\leq e^{4k_{y}^{2}T^{2}}\bigg(4k_{y}^{2}T^{2}\sup_{t\in\left[0,T\right]}\E\big[|\hat{y}_{t}-\bar{y}_{t}|^{2}\big] + 8\xi_{y}^{2}T\sup_{t\in\left[0,T\right]}\E\big[|\hat{y}_{t}-\bar{y}_{t}|\big] \nonumber\\[1pt]
&+ 8\xi_{y}^{2}T\sup_{t\in\left[0,T\right]}\E\big[|y_{t}-\hat{y}_{t}|\big]\bigg).
\end{align}
On the other hand, we know that
\begin{equation}\label{eq3.25.4}
\sup_{t\in\left[0,T\right]}\E\big[|y_{t}-\bar{y}_{t}|^{2}\big] \leq 2\sup_{t\in\left[0,T\right]}\E\left[|y_{t}-\hat{y}_{t}|^{2}\right] + 2\sup_{t\in\left[0,T\right]}\E\left[|\hat{y}_{t}-\bar{y}_{t}|^{2}\right].
\end{equation}
Substituting into \eqref{eq3.25.4} with the upper bound in \eqref{eq3.25.3}, we get
\begin{align}\label{eq3.25.5}
\sup_{t\in\left[0,T\right]}\E\big[|y_{t}-\bar{y}_{t}|^{2}\big] &\leq 2\Big(1+4k_{y}^{2}T^{2}e^{4k_{y}^{2}T^{2}}\Big)\sup_{t\in\left[0,T\right]}\E\big[|\hat{y}_{t}-\bar{y}_{t}|^{2}\big] + 16\xi_{y}^{2}Te^{4k_{y}^{2}T^{2}}\sup_{t\in\left[0,T\right]}\E\big[|\hat{y}_{t}-\bar{y}_{t}|\big] \nonumber\\[3pt]
&+ 16\xi_{y}^{2}Te^{4k_{y}^{2}T^{2}}\sup_{t\in\left[0,T\right]}\E\big[|y_{t}-\hat{y}_{t}|\big].
\end{align}
The convergence of the three terms on the right-hand side of \eqref{eq3.25.5} is a consequence of Lemma A.3 and Theorem 4.2 in \cite{Lord:2010}, which concludes the proof.%\hfill
\qquad
\end{proof}

%%%%%%%%%%%%%%%%%%%%%%%%%%%%%%%%%%%%%%%%%%%%%%%%%%%%%%%%%%%%%%%%%%%%%%%%%%%%%%
\subsection{The four-dimensional system}\label{subsec:system}

Even though weak convergence is important when estimating expectations of payoffs, strong convergence plays a crucial role in multilevel Monte Carlo methods and may be required for complex path-dependent derivatives. First, we prove the uniform mean square convergence of the stopped discretized spot FX rate process in \eqref{eq2.12}.

\begin{proposition}\label{Prop3.4.1} Let $L_{s}>S_{0}$, $L_{v}>v_{0}$, $L_{d}>g^{d}_{0}$, $L_{f}>g^{f}_{0}>\kappa^{-1}$, and define the stopping time
\begin{equation}\label{eq3.28}
\tau = \inf\!\big\{t\geq 0 :\,S_{t}\geq L_{s} \text{ or }\, v_{t}\geq L_{v} \text{ or }\, \hat{v}_{t}\geq L_{v} \text{ or }\, \hat{g}^{d}_{t}\geq L_{d} \text{ or }\, \hat{g}^{f}_{t}\geq L_{f} \text{ or }\, g^{f}_{t}\leq \kappa^{-1}\big\}.
\end{equation}
Under assumptions {\rm ($\mathcal{A}$1)} to {\rm ($\mathcal{A}$3)}, the stopped process converges uniformly in $L^{2}$, i.e.,
\begin{equation}\label{eq3.29}
\plim_{\delta t\to0}\E\bigg[\sup_{t\in[0,T]}\left|S_{t\wedge\tau}-\bar{S}_{t\wedge\tau}\right|^{2}\bigg] = 0.
\end{equation}
\end{proposition}
\begin{proof}
The absolute difference between the original and the discretized stopped processes can be bounded from above as follows,
\begin{align*}
\big|S_{t\wedge\tau}-\bar{S}_{t\wedge\tau}\big| &\leq \bigg|\int_{0}^{t\wedge\tau}{\hspace{-1pt}h(u)\big(S_{u}-\bar{S}_{u}\big)du}\,\bigg| + \bigg|\int_{0}^{t\wedge\tau}{\hspace{-1pt}\big(g^{d}_{u}-\bar{g}^{d}_{u}\big)S_{u}\,du}\,\bigg| + \bigg|\int_{0}^{t\wedge\tau}{\hspace{-1pt}\bar{g}^{d}_{u}\big(S_{u}-\bar{S}_{u}\big)du}\,\bigg| \\[2pt]
&\hspace{-2.5em} + \bigg|\int_{0}^{t\wedge\tau}{\hspace{-1pt}\big(g^{f}_{u}-\bar{g}^{f}_{u}\big)S_{u}\,du}\,\bigg|
+ \bigg|\int_{0}^{t\wedge\tau}{\hspace{-1pt}\bar{g}^{f}_{u}\big(S_{u}-\bar{S}_{u}\big)du}\,\bigg| + \bigg|\int_{0}^{t\wedge\tau}{\hspace{-1pt}\bar{\sigma}\big(u,\bar{S}_{u}\big)\sqrt{\bar{v}_{u}}\,\big(S_{u}-\bar{S}_{u}\big)dW_{u}^{s}}\,\bigg| \\[2pt]
&\hspace{-2.5em} + \bigg|\int_{0}^{t\wedge\tau}{\hspace{-1pt}\Big(\sigma\big(u,S_{u}\big)-\bar{\sigma}\big(u,\bar{S}_{u}\big)\Big)\sqrt{\bar{v}_{u}}\,S_{u}\,dW_{u}^{s}}\,\bigg| + \bigg|\int_{0}^{t\wedge\tau}{\hspace{-1pt}\sigma\big(u,S_{u}\big)\Big(\sqrt{v_{u}}-\sqrt{\bar{v}_{u}}\Big)S_{u}\,dW_{u}^{s}}\,\bigg|.
\end{align*}
Squaring both sides, taking the supremum over all $t \in [0,s]$, where $0 \leq s \leq T$, and then employing the Cauchy--Schwarz inequality leads to the upper bound
\begin{align}\label{eq3.29.1}
\sup_{t\in[0,s]}\big|S_{t\wedge\tau}-\bar{S}_{t\wedge\tau}\big|^{2} &\leq 8\big(L_{d}^{2}+L_{f}^{2}+4h_{max}^{2}\big)T\int_{0}^{s\wedge\tau}{\hspace{-1pt}\big(S_{u}-\bar{S}_{u}\big)^{2}du} + 8L_{s}^{2}T\int_{0}^{T}{\hspace{-1pt}\big(g^{d}_{u}-\bar{g}^{d}_{u}\big)^{2} du} \nonumber\\[-1pt]
&\hspace{-1em} + 8\sup_{t \in [0,s]}\bigg|\int_{0}^{t\wedge\tau}\hspace{-1pt}\sigma\big(u,S_{u}\big)\Big(\sqrt{v_{u}}-\sqrt{\bar{v}_{u}}\Big)S_{u}\,dW_{u}^{s}\,\bigg|^{2} + 16L_{s}^{2}T\int_{0}^{T}{\hspace{-1pt}\big(\hat{g}^{f}_{u}-\bar{g}^{f}_{u}\big)^{2} du} \nonumber\\[0pt]
&\hspace{-1em} + 8\sup_{t \in [0,s]}\bigg|\int_{0}^{t\wedge\tau}{\hspace{-1pt}\bar{\sigma}\big(u,\bar{S}_{u}\big)\sqrt{\bar{v}_{u}}\,\big(S_{u}-\bar{S}_{u}\big)dW_{u}^{s}}\,\bigg|^{2} + 16L_{s}^{2}T\int_{0}^{s\wedge\tau}{\hspace{-1pt}\big(g^{f}_{u}-\hat{g}^{f}_{u}\big)^{2} du} \nonumber\\[0pt]
&\hspace{-1em} + 8\sup_{t \in [0,s]}\bigg|\int_{0}^{t\wedge\tau}{\hspace{-1pt}\Big(\sigma\big(u,S_{u}\big)-\bar{\sigma}\big(u,\bar{S}_{u}\big)\Big)\sqrt{\bar{v}_{u}}\,S_{u}\,dW_{u}^{s}}\,\bigg|^{2}.
\end{align}
Note that
\begin{equation}\label{eq3.29.2}
\int_{0}^{s\wedge\tau}{\hspace{-1pt}\big(g^{f}_{u}-\hat{g}^{f}_{u}\big)^{2} du} \leq \int_{0}^{s}{\hspace{-1pt}\big(g^{f}_{u\wedge\tau}-\hat{g}^{f}_{u\wedge\tau}\big)^{2} du} \leq T\sup_{t\in[0,s]}\big|g^{f}_{t\wedge\tau}-\hat{g}^{f}_{t\wedge\tau}\big|^{2}.
\end{equation}
Hence, taking expectations in \eqref{eq3.29.1}, using Fubini's theorem, Doob's martingale inequality and the It\^o isometry, and upon noticing that a stopped martingale is also a martingale, we obtain
\begin{align}\label{eq3.30}
\E\bigg[\sup_{t\in[0,s]}\big|S_{t\wedge\tau}-\bar{S}_{t\wedge\tau}\big|^{2}\bigg] &\leq 8L_{s}^{2}T^{2}\sup_{t\in\left[0,T\right]}\E\left[\big|g^{d}_{t}-\bar{g}^{d}_{t}\big|^{2}\right] + 32\sigma_{max}^{2}L_{s}^{2}T\sup_{t\in\left[0,T\right]}\E\Big[\big|v_{t}-\bar{v}_{t}\big|\Big] \nonumber \\[1pt]
&\hspace{-8.5em} + 16L_{s}^{2}T^{2}\sup_{t\in\left[0,T\right]}\E\left[\big|\hat{g}^{f}_{t}-\bar{g}^{f}_{t}\big|^{2}\right] + \Big[8\big(L_{d}^{2}+L_{f}^{2}+4h_{max}^{2}\big)T+32\sigma_{max}^{2}L_{v}\Big]\E\bigg[\int_{0}^{s\wedge\tau}{\hspace{-3pt}\big|S_{u}-\bar{S}_{u}\big|^{2}du}\bigg] \nonumber \\[2pt]
&\hspace{-8.5em} + 32L_{s}^{2}L_{v}\E\bigg[\int_{0}^{s\wedge\tau}{\big|\sigma\big(u,S_{u}\big)-\bar{\sigma}\big(u,\bar{S}_{u}\big)\big|^{2}du}\bigg] + 16L_{s}^{2}T^{2}\E\bigg[\sup_{t\in[0,s]}\big|g^{f}_{t\wedge\tau}-\hat{g}^{f}_{t\wedge\tau}\big|^{2}\bigg].
\end{align}
First, we bound the last expectation on the right-hand side of \eqref{eq3.30} from above. From \eqref{eq3.22.1} and \eqref{eq3.23.1}, since $|g^{f}_{t\wedge\tau}-\hat{g}^{f}_{t\wedge\tau}|\leq|g^{f}_{t\wedge\tau}-\tilde{g}^{f}_{t\wedge\tau}|$, we get
\begin{align}\label{eq3.30.1}
\big|g^{f}_{t\wedge\tau}-\hat{g}^{f}_{t\wedge\tau}\big| &\leq \bigg|-k_{f}\int_{0}^{t\wedge\tau}{\hspace{-1pt}\big(g^{f}_{u}-\hat{g}^{f}_{u}\big)du} - k_{f}\int_{0}^{t\wedge\tau}{\hspace{-1pt}\big(\hat{g}^{f}_{u}-\bar{g}^{f}_{u}\big)du} + \xi_{f}\int_{0}^{t\wedge\tau}{\hspace{-1pt}\Big(\sqrt{g^{f}_{u}}-\sqrt{\hat{g}^{f}_{u}}\hspace{1.5pt}\Big)dW^{f}_{u}} \nonumber\\[2pt]
&\hspace{-1em}+ \xi_{f}\int_{0}^{t\wedge\tau}{\hspace{-1pt}\Big(\sqrt{\hat{g}^{f}_{u}}-\sqrt{\bar{g}^{f}_{u}}\hspace{1.5pt}\Big)dW^{f}_{u}} - \rho_{s\hspace{-.7pt}f}\xi_{f}\int_{0}^{t\wedge\tau}{\hspace{-1pt}\sigma\big(u,S_{u}\big)\sqrt{v_{u}}\hspace{1pt}\Big(\sqrt{g^{f}_{u}}-\sqrt{\hat{g}^{f}_{u}}\hspace{1.5pt}\Big)du} \nonumber\\[2pt]
&\hspace{-1em}- \rho_{s\hspace{-.7pt}f}\xi_{f}\int_{0}^{t\wedge\tau}{\hspace{-1pt}\sigma\big(u,S_{u}\big)\sqrt{v_{u}}\hspace{1pt}\Big(\sqrt{\hat{g}^{f}_{u}}-\sqrt{\bar{g}^{f}_{u}}\hspace{1.5pt}\Big)du} - \rho_{s\hspace{-.7pt}f}\xi_{f}\int_{0}^{t\wedge\tau}{\hspace{-1pt}\sigma\big(u,S_{u}\big)\sqrt{\bar{g}^{f}_{u}}\hspace{1pt}\Big(\sqrt{v_{u}}-\sqrt{\bar{v}_{u}}\hspace{1.5pt}\Big)du}\hspace{1pt} \nonumber\\[2pt]
&\hspace{-1em}- \rho_{s\hspace{-.7pt}f}\xi_{f}\int_{0}^{t\wedge\tau}{\hspace{-1pt}\sqrt{\bar{v}_{u}\bar{g}^{f}_{u}}\hspace{1pt}\Big(\sigma\big(u,S_{u}\big)-\bar{\sigma}\big(u,\bar{S}_{u}\big)\Big)du}\bigg|.
\end{align}
Squaring both sides, taking the supremum over all $t\in[0,s']$, where $0\leq s'\leq s$, and employing the Cauchy--Schwarz inequality, then taking expectations and using Doob's martingale inequality and Fubini's theorem, we deduce that
\begin{align*}
\E\bigg[\sup_{t\in[0,s']}\big|g^{f}_{t\wedge\tau}-\hat{g}^{f}_{t\wedge\tau}\big|^{2}\bigg] &\leq 8k_{f}^{2}T\int_{0}^{s'}{\E\Big[\big(g^{f}_{u}-\hat{g}^{f}_{u}\big)^{2}\Ind_{u<\hspace{1pt}\tau}\hspace{-1pt}\Big]du} + 8k_{f}^{2}T\int_{0}^{T}{\E\Big[\big(\hat{g}^{f}_{u}-\bar{g}^{f}_{u}\big)^{2}\Big]du} \\[1pt]
&\hspace{-7.5em}+ 32\xi_{f}^{2}\int_{0}^{s'}{\E\bigg[\Big(\sqrt{g^{f}_{u}}-\sqrt{\hat{g}^{f}_{u}}\hspace{1.5pt}\Big)^{2}\Ind_{u<\hspace{1pt}\tau}\hspace{-1pt}\bigg]du} + 32\xi_{f}^{2}\int_{0}^{T}{\E\Big[\big|\hat{g}^{f}_{u}-\bar{g}^{f}_{u}\big|\Big]du} \\[2pt]
&\hspace{-7.5em}+ 8\rho_{s\hspace{-.7pt}f}^{2}\xi_{f}^{2}\sigma_{max}^{2}L_{v}T\int_{0}^{s'}{\E\bigg[\Big(\sqrt{g^{f}_{u}}-\sqrt{\hat{g}^{f}_{u}}\hspace{1.5pt}\Big)^{2}\Ind_{u<\hspace{1pt}\tau}\hspace{-1pt}\bigg]du} + 8\rho_{s\hspace{-.7pt}f}^{2}\xi_{f}^{2}\sigma_{max}^{2}L_{v}T\int_{0}^{T}{\E\Big[\big|\hat{g}^{f}_{u}-\bar{g}^{f}_{u}\big|\hspace{-1pt}\Big]du} \\[2pt]
&\hspace{-7.5em}+ 8\rho_{s\hspace{-.7pt}f}^{2}\xi_{f}^{2}\sigma_{max}^{2}L_{f}T\int_{0}^{T}{\E\Big[\big|v_{u}-\bar{v}_{u}\big|\Big]du} + 8\rho_{s\hspace{-.7pt}f}^{2}\xi_{f}^{2}L_{v}L_{f}T\E\bigg[\int_{0}^{s'\hspace{-1pt}\wedge\tau}{\hspace{-1pt}\big|\sigma\big(u,S_{u}\big)-\bar{\sigma}\big(u,\bar{S}_{u}\big)\big|^{2}du}\bigg].
\end{align*}
However,
\begin{equation}\label{eq3.30.2}
\Big(\sqrt{g^{f}_{u}}-\sqrt{\hat{g}^{f}_{u}}\hspace{1.5pt}\Big)^{2}\Ind_{u<\hspace{1pt}\tau} \leq \Big(\sqrt{g^{f}_{u\wedge\tau}}-\sqrt{\hat{g}^{f}_{u\wedge\tau}}\hspace{1.5pt}\Big)^{2} \leq \kappa\big(g^{f}_{u\wedge\tau}-\hat{g}^{f}_{u\wedge\tau}\big)^{2},
\end{equation}
and hence, since $s'\leq s$,
\begin{align}\label{eq3.30.3}
\E\bigg[\sup_{t\in[0,s']}\big|g^{f}_{t\wedge\tau}-\hat{g}^{f}_{t\wedge\tau}\big|^{2}\bigg] &\leq \big(8k_{f}^{2}T+32\xi_{f}^{2}\kappa+8\rho_{s\hspace{-.7pt}f}^{2}\xi_{f}^{2}\sigma_{max}^{2}L_{v}\kappa T\big)\int_{0}^{s'}{\E\bigg[\sup_{t\in[0,u]}\big|g^{f}_{t\wedge\tau}-\hat{g}^{f}_{t\wedge\tau}\big|^{2}\bigg]du} \nonumber\\[1pt]
&\hspace{-7.5em}+ 8\rho_{s\hspace{-.7pt}f}^{2}\xi_{f}^{2}L_{v}L_{f}T\E\bigg[\int_{0}^{s\wedge\tau}{\hspace{-1pt}\big|\sigma\big(u,S_{u}\big)-\bar{\sigma}\big(u,\bar{S}_{u}\big)\big|^{2}du}\bigg] + 8\rho_{s\hspace{-.7pt}f}^{2}\xi_{f}^{2}\sigma_{max}^{2}L_{f}T^{2}\sup_{t\in[0,T]}\E\Big[\big|v_{t}-\bar{v}_{t}\big|\Big] \nonumber\\[3pt]
&\hspace{-7.5em}+ \big(32\xi_{f}^{2}T+8\rho_{s\hspace{-.7pt}f}^{2}\xi_{f}^{2}\sigma_{max}^{2}L_{v}T^{2}\big)\sup_{t\in[0,T]}\E\Big[\big|\hat{g}^{f}_{t}-\bar{g}^{f}_{t}\big|\Big] + 8k_{f}^{2}T^{2}\sup_{t\in[0,T]}\E\Big[\big(\hat{g}^{f}_{t}-\bar{g}^{f}_{t}\big)^{2}\Big].
\end{align}
Therefore, applying Gronwall's inequality to \eqref{eq3.30.3}, since $s\leq T$, we get
\begin{align}\label{eq3.30.4}
\E\bigg[\sup_{t\in[0,s]}\big|g^{f}_{t\wedge\tau}-\hat{g}^{f}_{t\wedge\tau}\big|^{2}\bigg] &\leq e^{\beta_{1}T}\bigg\{8\rho_{s\hspace{-.7pt}f}^{2}\xi_{f}^{2}L_{v}L_{f}T\E\bigg[\int_{0}^{s\wedge\tau}{\hspace{-1pt}\big|\sigma\big(u,S_{u}\big)-\bar{\sigma}\big(u,\bar{S}_{u}\big)\big|^{2}du}\bigg] \nonumber\\[3pt]
&\hspace{-2.5em}+ 8\rho_{s\hspace{-.7pt}f}^{2}\xi_{f}^{2}\sigma_{max}^{2}L_{f}T^{2}\sup_{t\in[0,T]}\E\Big[\big|v_{t}-\bar{v}_{t}\big|\Big] + 8k_{f}^{2}T^{2}\sup_{t\in[0,T]}\E\Big[\big(\hat{g}^{f}_{t}-\bar{g}^{f}_{t}\big)^{2}\Big] \nonumber\\[1pt]
&\hspace{-2.5em}+ \big(32\xi_{f}^{2}T+8\rho_{s\hspace{-.7pt}f}^{2}\xi_{f}^{2}\sigma_{max}^{2}L_{v}T^{2}\big)\sup_{t\in[0,T]}\E\Big[\big|\hat{g}^{f}_{t}-\bar{g}^{f}_{t}\big|\Big]\bigg\},
\end{align}
where
\begin{equation*}
\beta_{1} = 8k_{f}^{2}T+32\xi_{f}^{2}\kappa+8\rho_{s\hspace{-.7pt}f}^{2}\xi_{f}^{2}\sigma_{max}^{2}L_{v}\kappa T.
\end{equation*}
Second, assumption {\rm ($\mathcal{A}$2)} on the leverage function ensures that
\begin{align}\label{eq3.30.5}
\E\bigg[\int_{0}^{s\wedge\tau}{\hspace{-1pt}\big|\sigma\big(u,S_{u}\big)-\bar{\sigma}\big(u,\bar{S}_{u}\big)\big|^{2}du}\bigg] &\leq 3A^{2}T(\delta t)^{2\alpha}+ 3B^{2}\E\bigg[\int_{0}^{T}{\hspace{-1pt}\big|S_{t}-S_{\bar{t}}\big|^{2}\Ind_{t<\hspace{1pt}\tau}dt}\bigg] \nonumber\\[2pt]
&\hspace{0em} + 3B^{2}\E\bigg[\int_{0}^{s\wedge\tau}{\hspace{-1pt}\sup_{t \in [0,u]}\big|S_{t}-\bar{S}_{t}\big|^{2}du}\bigg],
\end{align}
where $\bar{t}=\delta t\floor*{\frac{t}{\delta t}}$. Furthermore, we know that
\begin{equation}\label{eq3.30.6}
\int_{0}^{s\wedge\tau}{\hspace{-1pt}\big|S_{u}-\bar{S}_{u}\big|^{2}du} \leq \int_{0}^{s\wedge\tau}{\hspace{-1pt}\sup_{t \in [0,u]}\left|S_{t}-\bar{S}_{t}\right|^{2}du} \leq \int_{0}^{s}{\sup_{t \in [0,u]}\left|S_{t\wedge\tau}-\bar{S}_{t\wedge\tau}\right|^{2}du}.
\end{equation}
Substituting back into \eqref{eq3.30} with \eqref{eq3.30.4} -- \eqref{eq3.30.6}, we deduce that
\begin{align}\label{eq3.30.7}
\E\bigg[\sup_{t\in[0,s]}\big|S_{t\wedge\tau}-\bar{S}_{t\wedge\tau}\big|^{2}\bigg] &\leq \beta_{2}\int_{0}^{s}{\E\bigg[\sup_{t\in[0,u]}\big|S_{t\wedge\tau}-\bar{S}_{t\wedge\tau}\big|^{2}\bigg]du} + \beta_{3}(\delta t)^{2\alpha} + \beta_{4}\sup_{t\in\left[0,T\right]}\E\Big[\big|v_{t}-\bar{v}_{t}\big|\Big] \nonumber\\[3pt]
&\hspace{-1em}+ \beta_{5}\sup_{t\in\left[0,T\right]}\E\left[\big|g^{d}_{t}-\bar{g}^{d}_{t}\big|^{2}\right] + \beta_{6}\sup_{t\in[0,T]}\E\Big[\big|\hat{g}^{f}_{t}-\bar{g}^{f}_{t}\big|\Big] + \beta_{7}\sup_{t\in\left[0,T\right]}\E\left[\big|\hat{g}^{f}_{t}-\bar{g}^{f}_{t}\big|^{2}\right] \nonumber\\[1pt]
&\hspace{-1em}+ \beta_{8}\int_{0}^{T}{\hspace{-1pt}\E\Big[\big|S_{t}-S_{\bar{t}}\big|^{2}\Ind_{t<\hspace{1pt}\tau}\Big]dt},
\end{align}
where the constants $\beta_{i}$, for $2\leq i\leq8$, are defined below.
\begin{gather*}
\beta_{2} = 32\sigma_{max}^{2}L_{v} + 8\big(L_{d}^{2}+L_{f}^{2}+4h_{max}^{2}\big)T + 96B^{2}L_{s}^{2}L_{v}\big(1+4\rho_{s\hspace{-.7pt}f}^{2}\xi_{f}^{2}L_{f}T^{3}e^{\beta_{1}T}\big), \nonumber\\[4pt]
\beta_{3} = 96A^{2}L_{s}^{2}L_{v}\big(1+4\rho_{s\hspace{-.7pt}f}^{2}\xi_{f}^{2}L_{f}T^{3}e^{\beta_{1}T}\big)T,
\hspace{4pt}\beta_{4} = 32\sigma_{max}^{2}L_{s}^{2}T\big(1+4\rho_{s\hspace{-.7pt}f}^{2}\xi_{f}^{2}L_{f}T^{3}e^{\beta_{1}T}\big), \nonumber\\[4pt]
\beta_{5} = 8L_{s}^{2}T^{2},
\hspace{4pt}\beta_{6} = 128\xi_{f}^{2}L_{s}^{2}\big(4+\rho_{s\hspace{-.7pt}f}^{2}\sigma_{max}^{2}L_{v}T\big)T^{3}e^{\beta_{1}T},
\hspace{4pt}\beta_{7} = 16L_{s}^{2}T^{2}\big(1+8k_{f}^{2}T^{2}e^{\beta_{1}T}\big), \nonumber\\[4pt]
\beta_{8} = 96B^{2}L_{s}^{2}L_{v}\big(1+4\rho_{s\hspace{-.7pt}f}^{2}\xi_{f}^{2}L_{f}T^{3}e^{\beta_{1}T}\big).
\end{gather*}
Therefore, applying Gronwall's inequality to \eqref{eq3.30.7}, we obtain
\begin{align}\label{eq3.30.8}
\E\bigg[\sup_{t\in[0,T]}\big|S_{t\wedge\tau}-\bar{S}_{t\wedge\tau}\big|^{2}\bigg] &\leq e^{\beta_{2}T}\bigg\{\beta_{3}(\delta t)^{2\alpha} + \beta_{4}\sup_{t\in\left[0,T\right]}\E\Big[\big|v_{t}-\bar{v}_{t}\big|\Big] + \beta_{5}\sup_{t\in\left[0,T\right]}\E\left[\big|g^{d}_{t}-\bar{g}^{d}_{t}\big|^{2}\right] \nonumber\\[0pt]
&\hspace{-7.5em}+ \beta_{6}\sup_{t\in[0,T]}\E\Big[\big|\hat{g}^{f}_{t}-\bar{g}^{f}_{t}\big|\Big] + \beta_{7}\sup_{t\in\left[0,T\right]}\E\left[\big|\hat{g}^{f}_{t}-\bar{g}^{f}_{t}\big|^{2}\right] + \beta_{8}\int_{0}^{T}{\hspace{-1pt}\E\Big[\big|S_{t}-S_{\bar{t}}\big|^{2}\Ind_{t<\hspace{1pt}\tau}\Big]dt}\bigg\}.
\end{align}
The convergence of the first term on the right-hand side as $\delta t\!\to\!0$ is trivial, whereas the convergence of the next two terms is a consequence of Proposition \ref{Prop3.3.5}. We now show the convergence of the $L^{2}$ difference between $\hat{g}^{f}$ and $\bar{g}^{f}$.

Suppose that $t\in[t_{n},t_{n+1})$. Since $|\hat{g}^{f}_{t}-\bar{g}^{f}_{t}|\leq|\tilde{g}^{f}_{t}-\tilde{g}^{f}_{t_{n}}|$ and using the fact that $2\sqrt{|ab|}\leq |a|+|b|$, we can bound the $L^{2}$ difference from above as follows:
\begin{align*}
\big|\hat{g}^{f}_{t}-\bar{g}^{f}_{t}\big|^{2} &\leq
\Big(k_{f}\theta_{f}\delta t + 0.5|\rho_{s\hspace{-.7pt}f}|\xi_{f}\sigma_{max}\delta t\hspace{1pt}\hat{v}_{t_{n}} + \big(k_{f}+ 0.5|\rho_{s\hspace{-.7pt}f}|\xi_{f}\sigma_{max}\big)\delta t\hspace{1pt}\hat{g}_{t_{n}}^{f} +\xi_{f}\sqrt{\hat{g}_{t_{n}}^{f}}\hspace{1pt}\big|W^{f}_{t}-W^{f}_{t_{n}}\big|\Big)^{\hspace{-1pt}2} \\[4pt]
&\hspace{-2.5em}\leq 4k_{f}^{2}\theta_{f}^{2}(\delta t)^{2} + \rho_{s\hspace{-.7pt}f}^{2}\xi_{f}^{2}\sigma_{max}^{2}(\delta t)^{2}\hat{v}_{t_{n}}^{2} + \big(2k_{f}+ |\rho_{s\hspace{-.7pt}f}|\xi_{f}\sigma_{max}\big)^{2}(\delta t)^{2}\big(\hat{g}_{t_{n}}^{f}\big)^{2} + 4\xi_{f}^{2}\hat{g}_{t_{n}}^{f}\big(W^{f}_{t}-W^{f}_{t_{n}}\big)^{2}.
\end{align*}
Hence,
\begin{align}\label{eq3.30.9}
\sup_{t\in\left[0,T\right]}\E\left[\big|\hat{g}^{f}_{t}-\bar{g}^{f}_{t}\big|^{2}\right] &\leq
4k_{f}^{2}\theta_{f}^{2}(\delta t)^{2} + \rho_{s\hspace{-.7pt}f}^{2}\xi_{f}^{2}\sigma_{max}^{2}(\delta t)^{2}\sup_{0\leq n\leq N}\E\big[\hat{v}_{t_{n}}^{2}\big] + 4\xi_{f}^{2}\delta t\sup_{0\leq n\leq N}\E\big[\hat{g}_{t_{n}}^{f}\big] \nonumber\\[1pt]
&\hspace{0em}+ \big(2k_{f}+ |\rho_{s\hspace{-.7pt}f}|\xi_{f}\sigma_{max}\big)^{2}(\delta t)^{2}\sup_{0\leq n\leq N}\E\Big[\big(\hat{g}_{t_{n}}^{f}\big)^{2}\Big].
\end{align}
The convergence of the term on the left-hand side follows from Propositions \ref{Prop3.3.3} and \ref{Prop3.3.4}.

Finally, the integrand in the last term in \eqref{eq3.30.8} can be bounded from above as follows:
\begin{align}\label{eq3.30.10}
\E\Big[\big|S_{t}-S_{\bar{t}}\big|^{2}\Ind_{t<\hspace{1pt}\tau}\Big] &= \E\left[\bigg|\int_{\bar{t}}^{t}{\big[g^{d}_{u}-g^{f}_{u}+h(u)\big]S_{u}\hspace{1pt}du} + \int_{\bar{t}}^{t}{\sigma(u,S_{u})\sqrt{v_{u}}\hspace{1pt}S_{u}\hspace{1pt}dW_{u}^{s}}\bigg|^{2}\Ind_{t<\hspace{1pt}\tau}\right] \nonumber\\[0pt]
&\hspace{-4.75em} \leq 4\delta t\E\left[\int_{\bar{t}}^{t}{\Big[\big(g^{d}_{u}\big)^{2}+\big(g^{f}_{u}\big)^{2}+h(u)^{2}\Big]S_{u}^{2}\Ind_{u<\hspace{1pt}\tau}du}\right] + 4\E\left[\bigg|\int_{\bar{t}}^{t}{\sigma(u,S_{u})\sqrt{v_{u}}\hspace{1pt}S_{u}\Ind_{u<\hspace{1pt}\tau}dW_{u}^{s}}\bigg|^{2}\right] \nonumber\\[2pt]
&\hspace{-4.75em} \leq 4L_{s}^{2}(\delta t)^{2}\bigg\{\sup_{u\in[0,T]}\E\left[\big(g^{d}_{u}\big)^{2}\right] + \sup_{u\in[0,T]}\E\left[\big(g^{f}_{u}\big)^{2}\right] + 4h_{max}^{2}\bigg\} + 4\sigma_{max}^{2}L_{s}^{2}L_{v}\delta t.
\end{align}
The right-hand side is independent of the time $t$ and tends to zero as $\delta t\!\to\!0$ from Propositions \ref{Prop3.3.3} and \ref{Prop3.3.4}, which concludes the proof. %\hfill
\qquad
\end{proof}

Next, we prove the uniform convergence in probability of the discretized spot FX rate process.

\begin{proposition}\label{Prop3.4.2} If $2k_{f}\theta_{f}>\xi_{f}^{2}$ and under assumptions {\rm ($\mathcal{A}$1)} to {\rm ($\mathcal{A}$3)}, $\bar{S}$ converges uniformly in probability, i.e.,
\begin{equation}\label{eq3.31}
\plim_{\delta t\to0} \Prob\bigg(\sup_{t\in[0,T]}\big|S_{t}-\bar{S}_{t}\big|>\epsilon\bigg) = 0 \,,\hspace{1em} \forall\hspace{.5pt}\epsilon>0.
\end{equation}
\end{proposition}
\begin{proof}
Fix $\epsilon>0$ and note that we have the following inclusion of events,
\begin{align*}
\bigg\{\omega : \sup_{t \in [0,T]}\big|S_{t}(\omega)-\bar{S}_{t}(\omega)\big| > \epsilon\bigg\} &\subseteq \bigg\{\omega : \sup_{t \in [0,T]}\big|S_{t}(\omega)-\bar{S}_{t}(\omega)\big| > \epsilon,\,\tau(\omega) \geq T\bigg\} \\[1pt]
&\hspace{.1em} \cup \bigg\{\omega : \sup_{t \in [0,T]}\big|S_{t}(\omega)-\bar{S}_{t}(\omega)\big| > \epsilon,\,\tau(\omega) < T\bigg\}.
\end{align*}
Hence,
\begin{equation}\label{eq3.31.1}
\bigg\{\sup_{t\in0,T]}\big|S_{t}-\bar{S}_{t}\big| > \epsilon\bigg\} \subseteq \bigg\{\sup_{t\in[0,T]}\big|S_{t\wedge\tau}-\bar{S}_{t\wedge\tau}\big| > \epsilon\bigg\} \cup \big\{\tau \hspace{-1pt}< T\big\}.
\end{equation}
In terms of probabilities of events, the previous inclusion implies that
\begin{align}\label{eq3.32}
\Prob\bigg(\sup_{t\in[0,T]}\big|S_{t}-\bar{S}_{t}\big| > \epsilon \bigg) \leq \Prob\bigg(\sup_{t\in[0,T]}\big|S_{t\wedge\tau}-\bar{S}_{t\wedge\tau}\big| > \epsilon \bigg) + \Prob\big(\tau \hspace{-1pt}< T\big).
\end{align}
The convergence in probability of the stopped process is an immediate consequence of Proposition \ref{Prop3.4.1} and Markov's inequality. Moreover, from the definition of the stopping time in \eqref{eq3.28}, we get
\begin{align}\label{eq3.33}
\big\{\tau \hspace{-1pt}< T\big\} &\subseteq \bigg\{\sup_{t\in[0,T]}v_{t}\geq L_{v}\bigg\} \cup \bigg\{\sup_{t\in[0,T]}\hat{v}_{t}\geq L_{v}\bigg\} \cup \bigg\{\sup_{t\in[0,T]}\hat{g}^{d}_{t}\geq L_{d}\bigg\} \cup \bigg\{\sup_{t\in[0,T]}\hat{g}^{f}_{t}\geq L_{f}\bigg\} \nonumber \\[1pt]
&\cup \bigg\{\sup_{t\in[0,T]}S_{t}\geq L_{s}\bigg\} \cup \big\{\tau_{\kappa} < T\big\},
\end{align}
where we define
\begin{equation}\label{eq3.33.1}
\tau_{\kappa} = \inf\!\big\{t\geq 0 :\, g^{f}_{t}\leq\kappa^{-1}\big\}.
\end{equation}
However, assuming that $L_{s}>1$, we have that
\begin{equation}\label{eq3.33.2}
\bigg\{\sup_{t\in[0,T]}S_{t}\geq L_{s}\bigg\} = \bigg\{\sup_{t\in[0,T]}x_{t} \geq \log L_{s}\bigg\} \subseteq \bigg\{\sup_{t\in[0,T]}|x_{t}| \geq \log L_{s}\bigg\}.
\end{equation}
Therefore, using \eqref{eq3.33}, \eqref{eq3.33.2} and Markov's inequality, we find an upper bound
\begin{align}\label{eq3.33.3}
\Prob\big(\tau \hspace{-1pt}< T\big) &\leq L_{v}^{-1}\E\bigg[\sup_{t\in[0,T]}v_{t}\bigg] + L_{v}^{-1}\E\bigg[\sup_{t\in[0,T]}\hat{v}_{t}\bigg] + L_{d}^{-1}\E\bigg[\sup_{t\in[0,T]}\hat{g}^{d}_{t}\bigg] + L_{f}^{-1}\E\bigg[\sup_{t\in[0,T]}\hat{g}^{f}_{t}\bigg] \nonumber\\[1pt]
&+ \big(\log L_{s}\big)^{-1}\E\bigg[\sup_{t\in[0,T]}|x_{t}|\bigg] + \Prob\big(\tau_{\kappa}<T\big).
\end{align}
First, recall the formula for the logarithm of the spot FX rate process,
\begin{align}\label{eq3.33.4}
x_{t} = x_{0} + \int_{0}^{t}{\Big(g^{d}_{u} - g^{f}_{u} + h(u) - \frac{1}{2}\sigma^{2}\big(u,S_{u}\big)v_{u}\Big) du} + \int_{0}^{t}{\sigma\big(u,S_{u}\big)\sqrt{v_{u}}\,dW^{s}_{u}}.
\end{align}
Taking the supremum over all $t\in[0,T]$ and using the fact that $a<a^{2}+1$, then taking expectations and employing Doob's martingale inequality and the It\^o isometry, we get
\begin{equation}\label{eq3.33.5}
\E\!\bigg[\sup_{t\in[0,T]}|x_{t}|\bigg] \leq 1 + |x_{0}| + 2h_{max}T + T\!\E\!\bigg[\sup_{t\in[0,T]}g^{d}_{t}\bigg] + T\!\E\!\bigg[\sup_{t\in[0,T]}g^{f}_{t}\bigg] + 4.5\sigma_{max}^{2}T\!\E\!\bigg[\sup_{t\in[0,T]}v_{t}\bigg].
\end{equation}
Second, we prove the almost sure positivity of the process $g^{f}$ under the Feller condition, i.e., when $2k_{f}\theta_{f}>\xi_{f}^{2}$. To this end, let $\alpha=(2k_{f}\theta_{f}-\xi_{f}^{2})/2\xi_{f}^{2}$ and define the function $F:(0,\infty)\mapsto\RR$ by $F(x) = x^{-\alpha}$. Using It\^o's formula, we have that
\begin{align}\label{eq3.33.6}
\E\Big[F\big(g_{T\wedge\hspace{1pt}\tau_{\kappa}}^{f}\big)\Big] &= F\big(g_{0}^{f}\big) - \E\int_{0}^{T\wedge\hspace{1pt}\tau_{\kappa}}{\hspace{-1pt}\alpha\big(g_{u}^{f}\big)^{-(1+\alpha)}\big(k_{f}\theta_{f}-k_{f}g^{f}_{u}-\rho_{s\hspace{-.7pt}f}\xi_{f}\sigma\big(u,S_{u}\big)\sqrt{v_{u}g^{f}_{u}}\hspace{1pt}\big)du} \nonumber\\[2pt]
&\hspace{-1em}+ \frac{1}{2}\hspace{.5pt}\E\int_{0}^{T\wedge\hspace{1pt}\tau_{\kappa}}{\hspace{-1pt}\alpha(1+\alpha)\xi_{f}^{2}\big(g_{u}^{f}\big)^{-(1+\alpha)}du}
- \E\int_{0}^{T\wedge\hspace{1pt}\tau_{\kappa}}{\hspace{-1pt}\alpha\xi_{f}\big(g_{u}^{f}\big)^{-(0.5+\alpha)}dW^{f}_{u}}.
\end{align}
Note that
\begin{equation*}
\alpha k_{f}\theta_{f} - \frac{1}{2}\hspace{1pt}\alpha(1+\alpha)\xi_{f}^{2} = \frac{1}{4}\hspace{1pt}\alpha\big[4k_{f}\theta_{f}-2\xi_{f}^{2}-2\alpha\xi_{f}^{2}\big] = \frac{1}{2}\hspace{1pt}\alpha^{2}\xi_{f}^{2},
\end{equation*}
and also that
\begin{equation*}
\E\int_{0}^{T}{\alpha^{2}\xi_{f}^{2}\big(g_{u}^{f}\big)^{-(1+2\alpha)}\Ind_{u<\hspace{.5pt}\tau_{\kappa}}du} \leq \alpha^{2}\xi_{f}^{2}\kappa^{1+2\alpha}\hspace{.5pt}T < \infty,
\end{equation*}
so the stochastic integral on the right-hand side of \eqref{eq3.33.6} is a true martingale. Hence,
\begin{align}\label{eq3.33.7}
\E\Big[F\big(g_{T\wedge\hspace{1pt}\tau_{\kappa}}^{f}\big)\Big] &\leq F\big(g_{0}^{f}\big) - \frac{1}{2}\hspace{1pt}\alpha^{2}\xi_{f}^{2}\E\int_{0}^{T}{\big(g_{u}^{f}\big)^{-(1+\alpha)}\Ind_{u<\hspace{.5pt}\tau_{\kappa}}du} + \alpha k_{f}\E\int_{0}^{T}{\big(g_{u}^{f}\big)^{-\alpha}\Ind_{u<\hspace{.5pt}\tau_{\kappa}}du} \nonumber\\[1pt]
&+ \alpha|\rho_{s\hspace{-.7pt}f}|\xi_{f}\sigma_{max}\E\int_{0}^{T}{v_{u}^{0.5}\big(g_{u}^{f}\big)^{-(0.5+\alpha)}\Ind_{u<\hspace{.5pt}\tau_{\kappa}}du}.
\end{align}
Employing Fubini's theorem and H\"older's inequality, we deduce that
\begin{align}\label{eq3.33.8}
\E\Big[F\big(g_{T\wedge\hspace{1pt}\tau_{\kappa}}^{f}\big)\Big] &\leq F\big(g_{0}^{f}\big) - \frac{1}{2}\hspace{1pt}\alpha^{2}\xi_{f}^{2}\int_{0}^{T}\hspace{-1.5pt}\bigg(\E\hspace{-1pt}\Big[\big(g_{u}^{f}\big)^{-(1+\alpha)}\Ind_{u<\hspace{.5pt}\tau_{\kappa}}\Big] - 2\alpha^{-1}k_{f}\xi_{f}^{-2}\E\hspace{-1pt}\Big[\big(g_{u}^{f}\big)^{-(1+\alpha)}\Ind_{u<\hspace{.5pt}\tau_{\kappa}}\Big]^{\frac{\alpha}{1+\alpha}} \nonumber\\[3pt]
&- 2\alpha^{-1}|\rho_{s\hspace{-.7pt}f}|\xi_{f}^{-1}\sigma_{max}\sup_{t\in[0,T]}\E\hspace{-1pt}\Big[v_{t}^{1+\alpha}\Big]^{\frac{1}{2(1+\alpha)}}\E\hspace{-1pt}\Big[\big(g_{u}^{f}\big)^{-(1+\alpha)}\Ind_{u<\hspace{.5pt}\tau_{\kappa}}\Big]^{\frac{1+2\alpha}{2(1+\alpha)}}\bigg)du.
\end{align}
However, the moments of the square-root process are bounded by Proposition \ref{Prop3.3.3}. Furthermore, if $p,q\in(0,1)$, $c_{p,q}\geq0$ and $c=c_{p}+c_{q}>0$, then for all $x\geq0$ we have
\begin{align*}
x - c_{p}x^{p} - c_{q}x^{q} &= \frac{c_{p}}{c}\big(x-cx^{p}\big) + \frac{c_{q}}{c}\big(x-cx^{q}\big) \nonumber\\[2pt]
&= c_{p}c^{\frac{p}{1-p}}\Big[\hspace{.5pt}c^{-\frac{1}{1-p}}x-\Big(c^{-\frac{1}{1-p}}x\Big)^{p}\hspace{1pt}\Big] + c_{q}c^{\frac{q}{1-q}}\Big[\hspace{.5pt}c^{-\frac{1}{1-q}}x-\Big(c^{-\frac{1}{1-q}}x\Big)^{q}\hspace{1pt}\Big] \nonumber\\[2pt]
&\geq - c_{p}c^{\frac{p}{1-p}} - c_{q}c^{\frac{q}{1-q}}.
\end{align*}
Therefore, the integrand in \eqref{eq3.33.8} is bounded from below by a constant and thus
\begin{equation}\label{eq3.33.9}
\E\Big[F\big(g_{T\wedge\hspace{1pt}\tau_{\kappa}}^{f}\big)\Big] \leq C,
\end{equation}
for some constant $C$ independent of $\kappa$. Since $g^{f}$ has continuous paths, $g_{\tau_{\kappa}}^{f}=\kappa^{-1}$ and $F(g_{\tau_{\kappa}}^{f})=\kappa^{\alpha}$. Hence, from \eqref{eq3.33.9} and the positivity of $F$, we deduce that
\begin{equation}\label{eq3.33.10}
\Prob\big(\tau_{\kappa}<T\big) = \kappa^{-\alpha}\E\Big[F\big(g_{\tau_{\kappa}}^{f}\big)\Ind_{\tau_{\kappa}<\hspace{1pt}T}\Big] \leq \kappa^{-\alpha}\E\Big[F\big(g_{T\wedge\hspace{1pt}\tau_{\kappa}}^{f}\big)\Big] \leq C\kappa^{-\alpha}.
\end{equation}
Taking the limit as $\delta t\!\to\!0$ in \eqref{eq3.32}, using the upper bounds derived in \eqref{eq3.33.3}, \eqref{eq3.33.5} and \eqref{eq3.33.10}, and then employing Propositions \ref{Prop3.3.3} and \ref{Prop3.3.4}, the conclusion follows from the fact that we can choose $L_{s}$, $L_{v}$, $L_{d}$, $L_{f}$ and $\kappa$ arbitrarily large.%\hfill
\qquad
\end{proof}

Many models with stochastic volatility dynamics, including the Heston model, have the undesirable feature of moment instability, i.e., moments of order higher than 1 can explode in finite time \cite{Andersen:2007}. This can cause problems in practice, for instance when computing the arbitrage-free price of an option with a superlinear payoff. Furthermore, establishing the existence of moments of order higher than 1 of the process and its approximation is an important ingredient in the convergence analysis (see \cite{Higham:2002}). Since the most popular FX contracts grow at most linearly in the FX rate and their risk-neutral valuation involves computing the expected discounted payoff, it is useful to study the finiteness of moments under discounting. Define $R$ to be the discounted exchange rate process,
\begin{align}\label{eq3.34.2}
R_{t} = S_{0}\exp\bigg\{\hspace{-1pt}-\int_{0}^{t}{r^{f}_{u}\,du} - \frac{1}{2}\int_{0}^{t}{\sigma^{2}(u,S_{u})v_{u}\,du} + \int_{0}^{t}{\sigma(u,S_{u})\sqrt{v_{u}}\,dW_{u}^{s}}\bigg\}\hspace{.5pt},
\end{align}
and let $\bar{R}$ be its continuous-time approximation,
\begin{align}\label{eq3.34.3}
\bar{R}_{t} = S_{0}\exp\bigg\{\hspace{-1pt}-\int_{0}^{t}{\bar{r}^{f}_{u}\,du} - \frac{1}{2}\int_{0}^{t}{\bar{\sigma}^{2}\big(u,\bar{S}_{u}\big)\bar{v}_{u}\,du} + \int_{0}^{t}{\bar{\sigma}\big(u,\bar{S}_{u}\big)\sqrt{\bar{v}_{u}}\,dW^{s}_{u}}\bigg\}\hspace{.5pt}.
\end{align}
The next result establishes a lower bound on the explosion time of moments of order higher than 1 of the discounted process and is an important step in proving Theorem \ref{Thm2.1} (Theorem 1 in \cite{Mariapragassam:2016}), which, in turn, plays a key role in the calibration of the model.

\begin{proposition}\label{Prop3.4.3} Let $\alpha\geq1$. Under assumptions {\rm ($\mathcal{A}$1)} and {\rm ($\mathcal{A}$3)}, if $T<T^{*}$, there exists $\omega_{1}>\alpha$ such that, for all $\omega \in [1,\omega_{1})$, the following holds:
\begin{equation}\label{eq3.35}
\sup_{t\in[0,T]}\E\big[R_{t}^{\hspace{1pt}\omega}\big] < \infty,
\end{equation}
where $\varphi(\alpha)=\alpha+\sqrt{(\alpha-1)\alpha}\hspace{1pt}$, $\zeta=\xi\sigma_{max}$ and $T^{*}$ is as given below.
\begin{enumerate}[(1)]
\item{When $k<\varphi(\alpha)\zeta$,
\begin{equation}\label{eq3.35.1}
T^{*} = \frac{2}{\sqrt{\varphi(\alpha)^{2}\zeta^{2}-k^{2}}}\bigg[\frac{\pi}{2}+\arctan\bigg(\frac{k}{\sqrt{\varphi(\alpha)^{2}\zeta^{2}-k^{2}}}\bigg)\bigg].
\end{equation}}
\item{When $k\geq\varphi(\alpha)\zeta$,
\begin{equation}\label{eq3.35.2}
T^{*} = \infty\hspace{.5pt}.
\end{equation}}
\end{enumerate}
\end{proposition}
\begin{proof}
First, we find it convenient to define a new stochastic process $L$ by
\begin{equation}\label{eq3.36}
L_{t} \equiv S_{0}\exp\left\{h_{max}t - \hspace{1pt}\frac{1}{2}\int_{0}^{t}{\sigma^{2}(u,S_{u})v_{u}\,du} + \int_{0}^{t}{\sigma(u,S_{u})\sqrt{v_{u}}\,dW_{u}^{s}}\right\}.
\end{equation}
As $R_{t}\leq L_{t}$ for all $t \in [0,T]$, it suffices to prove the finiteness of the supremum over $t$ of
\begin{align}\label{eq3.37}
\E\big[L_{t}^{\omega}\big] &= S_{0}^{\omega}\E\bigg[\exp\bigg\{\omega h_{max}t + \omega\int_{0}^{t}{\sigma(u,S_{u})\sqrt{v_{u}}\,dW_{u}^{s}} - \frac{\omega}{2}\int_{0}^{t}{\sigma^{2}(u,S_{u})v_{u}\,du}\bigg\}\bigg].
\end{align}
Second, suppose that $T<T^{*}$, with $T^{*}$ defined in \eqref{eq3.35.1} -- \eqref{eq3.35.2}. If $k<\varphi(\alpha)\zeta$, since $\varphi(\cdot)$ is increasing on $[1,\infty)$ and by a continuity argument, we can find $\omega_{1}>\alpha$ such that, for all $\omega\in(\alpha,\omega_{1})$,
\begin{equation}\label{eq3.37.0.1}
k<\varphi(\omega)\zeta \hspace{5pt}\text{ and }\hspace{5pt} T<\frac{2}{\sqrt{\varphi(\omega)^{2}\zeta^{2}-k^{2}}}\bigg[\frac{\pi}{2}+\arctan\bigg(\frac{k}{\sqrt{\varphi(\omega)^{2}\zeta^{2}-k^{2}}}\bigg)\bigg].
\end{equation}
If $k=\varphi(\alpha)\zeta$, since $\varphi(\cdot)$ is strictly increasing on $[1,\infty)$ and
\begin{equation*}
\lim_{\omega\hspace{1pt}\downarrow\hspace{1pt}\alpha^{+}}\hspace{1pt}\frac{2}{\sqrt{\varphi(\omega)^{2}\zeta^{2}-k^{2}}}\bigg[\frac{\pi}{2}+\arctan\bigg(\frac{k}{\sqrt{\varphi(\omega)^{2}\zeta^{2}-k^{2}}}\bigg)\bigg] = \infty,
\end{equation*}
we can find $\omega_{1}>\alpha$ such that \eqref{eq3.37.0.1} holds for all $\omega\in(\alpha,\omega_{1})$. Finally, if $k>\varphi(\alpha)\zeta$, by a continuity argument, we can find $\omega_{1}>\alpha$ such that, for all $\omega\in(\alpha,\omega_{1})$,
\begin{equation}\label{eq3.37.0.3}
k>\varphi(\omega)\zeta.
\end{equation}
Third, fix $\omega\in(\alpha,\omega_{1})$ and let the functions $f_{\omega},g_{\omega} : (1,\infty) \mapsto (1,\infty)$ be defined by
\begin{equation*}
f_{\omega}(x) = \omega^{2}x^{2}
\end{equation*}
and
\begin{equation*}
g_{\omega}(x) = \frac{\omega x}{x-1}\big(\omega x-1\big) = \big(\omega+\sqrt{(\omega-1)\omega}\hspace{1pt}\big)^{2} + \frac{\omega^{2}}{x-1}\big(x-1-\sqrt{1-1/\omega}\hspace{1pt}\big)^{2}.
\end{equation*}
The first function is strictly increasing, whereas the second function attains a global minimum at $x^{*}=1+\sqrt{1-1/\omega}\hspace{1pt}$. Furthermore, they are equal at $x^{*}$, so
\begin{equation}\label{eq3.38.1}
\inf_{x>1}\,\max\big\{f_{\omega}(x)\hspace{1pt},\hspace{1pt} g_{\omega}(x)\big\} = g_{\omega}(x^{*}) = \big(\omega+\sqrt{(\omega-1)\omega}\hspace{1pt}\big)^{2}.
\end{equation}
Consider the H\"older pair $(p,q)$ satisfying $q=p/(p-1)$ such that
\begin{equation}\label{eq3.38.2}
p = 1+\sqrt{\frac{\omega-1}{\omega}} \hspace{5pt}\text{ and }\hspace{5pt} q = 1+\sqrt{\frac{\omega}{\omega-1}}\hspace{1pt}.
\end{equation}
This particular H\"older pair ensures an optimal lower bound $T^{*}$ on the explosion time. Next, define the quantity $a = p\omega^{2}-\omega$ and introduce the stochastic process
\begin{equation*}
M_{t} = p\omega\int_{0}^{t}{\sigma(u,S_{u})\sqrt{v_{u}}\,dW_{u}^{s}}
\end{equation*}
with quadratic variation
\begin{equation*}
\langle M\rangle_{t} = p^{2}\omega^{2}\int_{0}^{t}{\sigma^{2}(u,S_{u})v_{u}\,du}.
\end{equation*}
Then we can rewrite \eqref{eq3.37} as follows:
\begin{align}\label{eq3.39.1}
\E\big[L_{t}^{\omega}\big] = S_{0}^{\omega}e^{\omega h_{max}t}\E\bigg[\exp\bigg\{\frac{1}{p}\left[M_{t} - \frac{1}{2}\langle M\rangle_{t}\right] + \frac{a}{2}\int_{0}^{t}{\sigma^{2}(u,S_{u})v_{u}\,du}\bigg\}\bigg].
\end{align}
Applying H\"older's inequality with the pair $(p,q)$ from \eqref{eq3.38.2} and taking the supremum over $[0,T]$,
\begin{align}\label{eq3.39.2}
\sup_{t \in [0,T]}\E\big[L_{t}^{\omega}\big] &\leq S_{0}^{\omega}e^{\omega h_{max}T}\sup_{t \in [0,T]}\E\bigg[\exp\bigg\{M_{t} - \frac{1}{2}\langle M\rangle_{t}\bigg\}\bigg]^{\frac{1}{p}} \nonumber \\[0pt]
&\times \E\bigg[\exp\bigg\{\frac{1}{2}\hspace{1pt}q\omega\big(p\omega-1\big)\sigma_{max}^{2}\int_{0}^{T}{v_{u}\,du}\bigg\}\bigg]^{\frac{1}{q}}.
\end{align}
The stochastic exponential is a martingale if Novikov's condition is satisfied, i.e.,
\begin{align*}
\E\bigg[\exp\bigg\{\,\frac{1}{2}\langle M\rangle_{T}\bigg\}\bigg] \leq \E\bigg[\exp\bigg\{\frac{1}{2}\hspace{1pt}p^{2}\omega^{2}\sigma_{max}^{2}\int_{0}^{T}{v_{u}\,du}\bigg\}\bigg] < \infty.
\end{align*}
The finiteness of the two expectations in \eqref{eq3.39.2} follows from \eqref{eq3.37.0.1} -- \eqref{eq3.38.2} and Proposition \ref{Prop3.3.1}. The extension to the interval $[1,\omega_{1})$ follows from Jensen's inequality.%\hfill
\qquad
\end{proof}

\begin{proposition}\label{Prop3.4.4} Let $\alpha\geq1$. Under assumptions {\rm ($\mathcal{A}$1)} and {\rm ($\mathcal{A}$3)}, if $T<T^{*}$, there exists $\omega_{2}>\alpha$ such that, for all $\omega \in [1,\omega_{2})$ and $\delta_{T}<k^{-1}$, the following holds:
\begin{equation}\label{eq3.40}
\sup_{\delta t \in (0,\delta_{T})}\hspace{1.5pt}\sup_{t \in [0,T]} \E\big[(\bar{R}_{t})^{\hspace{.5pt}\omega}\big] < \infty,
\end{equation}
where $\varphi(\alpha)=\alpha+\sqrt{(\alpha-1)\alpha}\hspace{1pt}$, $\zeta=\xi\sigma_{max}$ and $T^{*}$ is as given below.
\begin{enumerate}[(1)]
\item{When $k\leq\frac{1}{2}\hspace{1pt}\varphi(\alpha)\zeta$,
\begin{equation}\label{eq3.40.1}
T^{*} = \frac{1}{\varphi(\alpha)\zeta-k}\hspace{1pt}.
\end{equation}}
\item{When $k>\frac{1}{2}\hspace{1pt}\varphi(\alpha)\zeta$,
\begin{equation}\label{eq3.40.2}
T^{*} = \frac{4k}{\varphi(\alpha)^{2}\zeta^{2}}\hspace{1pt}.
\end{equation}}
\end{enumerate}
\end{proposition}
\begin{proof}
For convenience, define a new stochastic process $\bar{L}$ by
\begin{align}\label{eq3.41}
\bar{L}_{t} &\equiv S_{0}\exp\bigg\{h_{max}t - \hspace{1pt}\frac{1}{2}\int_{0}^{t}{\bar{\sigma}^{2}\big(u,\bar{S}_{u}\big)\bar{v}_{u}\,du} + \int_{0}^{t}{\bar{\sigma}\big(u,\bar{S}_{u}\big)\sqrt{\bar{v}_{u}}\,dW^{s}_{u}}\bigg\}.
\end{align}
Since $\bar{R}_{t} \leq \bar{L}_{t}$ for all $t \in [0,T]$, it suffices to prove the finiteness of the supremum over $t$ and $\delta t$ of
\begin{align}\label{eq3.42}
\E\big[(\bar{L}_{t})^{\omega}\big] = S_{0}^{\omega}\E\bigg[\exp\bigg\{\omega h_{max}t + \omega\int_{0}^{t}{\bar{\sigma}\big(u,\bar{S}_{u}\big)\sqrt{\bar{v}_{u}}\,dW^{s}_{u}} - \frac{\omega}{2}\int_{0}^{t}{\bar{\sigma}^{2}\big(u,\bar{S}_{u}\big)\bar{v}_{u}\,du}\bigg\}\bigg].
\end{align}
Suppose that $T<T^{*}$, with $T^{*}$ from \eqref{eq3.40.1} -- \eqref{eq3.40.2}. If $k\leq\frac{1}{2}\hspace{1pt}\varphi(\alpha)\zeta$, by a continuity argument, we can find $\omega_{2}>\alpha$ such that, for all $\omega\in(\alpha,\omega_{2})$,
\begin{equation}\label{eq3.42.1}
k<\frac{1}{2}\hspace{1pt}\varphi(\omega)\zeta \hspace{5pt}\text{ and }\hspace{5pt} T<\frac{1}{\varphi(\omega)\zeta-k}\hspace{1pt}.
\end{equation}
On the other hand, if $k>\frac{1}{2}\hspace{1pt}\varphi(\alpha)\zeta$, by a continuity argument, we can find $\omega_{2}>\alpha$ such that, for all $\omega\in(\alpha,\omega_{2})$,
\begin{equation}\label{eq3.42.2}
k>\frac{1}{2}\hspace{1pt}\varphi(\omega)\zeta \hspace{5pt}\text{ and }\hspace{5pt} T<\frac{4k}{\varphi(\omega)^{2}\zeta^{2}}\hspace{1pt}.
\end{equation}
Henceforth, we argue as in Proposition \ref{Prop3.4.3} and use Theorem \ref{Thm3.3.2} to deduce the finiteness of the supremum over $t$ and $\delta t$ of \eqref{eq3.42}.
\qquad
\end{proof}

Higham et al. \cite{Higham:2002} proved that for a locally Lipschitz SDE, the boundedness of the $p$th moments of the exact and numerical solutions, for some $p>2$, ensures the strong mean square convergence of the Euler--Maruyama method. The existence of moment bounds for explicit Euler approximations, however, remained an open problem. Recently, Hutzenthaler et al. \cite{Hutzenthaler:2011} studied SDEs with superlinearly growing coefficients and proved strong and weak divergence in $L^{p}$ for all $p\geq1$, and hence that the moment-bound assumption is not satisfied. To the best of our knowledge, the uniform boundedness of moments of order greater than 1 of discretization schemes for the Heston model and extensions thereof has not yet been established -- this gap in the literature was also identified in Kloeden and Neuenkirch \cite{Kloeden:2012} -- and our Proposition \ref{Prop3.4.4} is the first result to address this issue.

Since the typical payoff of a FX contract grows at most linearly in the exchange rate, it suffices to know the strong convergence of the discounted process in $L^{1}$ to deduce the convergence of the time-discretization error to zero. The following theorem can be generalized to the $L^{\alpha}$ case relatively easily for all $\alpha\geq1$, upon noticing that the critical time $T^{*}$ from \eqref{eq3.35.1} -- \eqref{eq3.35.2} is always greater than the one from \eqref{eq3.40.1} -- \eqref{eq3.40.2}.

\begin{theorem}\label{Thm3.4.6} Under assumptions {\rm ($\mathcal{A}$1)} to {\rm ($\mathcal{A}$3)}, if $2k_{f}\theta_{f}>\xi_{f}^{2}$ and $T<T^{*}$, where $\zeta=\xi\sigma_{max}$ and
\begin{equation}\label{eq3.43}
T^{*} \equiv \hspace{1pt}\frac{4k}{\zeta^{2}}\hspace{.5pt}\Ind_{\zeta<\hspace{.5pt}2k} \hspace{1pt}+\hspace{3pt} \frac{1}{\zeta-k}\hspace{.5pt}\Ind_{\zeta\geq\hspace{.5pt}2k}\hspace{1pt},
\end{equation}
the discounted process converges strongly in $L^{1}$, i.e.,
\begin{equation}\label{eq3.45}
\plim_{\delta t \to 0}\hspace{1.5pt}\sup_{t \in [0,T]} \E\Big[\big|R_{t}-\bar{R}_{t}\big|\Big] = 0.
\end{equation}
\end{theorem}
\begin{proof}
Fix $\epsilon > 0$ and define the event $A = \Big\{\big|R_{t}-\bar{R}_{t}\big| > \epsilon\Big\}$, then
\begin{align*}
\sup_{t \in [0,T]} \E\Big[\big|R_{t}-\bar{R}_{t}\big|\Big] \leq \sup_{t \in [0,T]} \E\Big[\big|R_{t}-\bar{R}_{t}\big| \Ind_{A^{c}}\Big] + \sup_{t \in [0,T]} \E\Big[\big|R_{t}-\bar{R}_{t}\big| \Ind_{A}\Big].
\end{align*}
Hence,
\begin{align*}
\sup_{t \in [0,T]} \E\Big[\big|R_{t}-\bar{R}_{t}\big|\Big] \leq \epsilon + \sup_{t \in [0,T]} \E\big[R_{t}\Ind_{A}\big] + \sup_{t \in [0,T]} \E\big[\hspace{.5pt}\bar{R}_{t}\Ind_{A}\big].
\end{align*}
Choosing some $1 \hspace{-1.25pt}<\hspace{-1pt} \omega \hspace{-1pt}<\hspace{-1.1pt} \min\left\{\omega_{1},\omega_{2}\right\}$ and applying H\"older's inequality to the two expectations on the right-hand side with the pair $(p,q) = \big(\omega,\frac{\omega}{\omega-1}\big)$ returns the following upper bound:
\begin{equation*}
\sup_{t\in[0,T]}\E\Big[\big|R_{t}-\bar{R}_{t}\big|\Big] \leq \epsilon + \bigg\{\sup_{t\in[0,T]}\E\big[R_{t}^{\hspace{1pt}\omega}\big]^{\frac{1}{\omega}} + \sup_{t\in[0,T]}\E\big[(\bar{R}_{t})^{\hspace{.5pt}\omega}\big]^{\frac{1}{\omega}}\bigg\} \sup_{t\in[0,T]}\Prob\Big(\big|R_{t}-\bar{R}_{t}\big|>\epsilon\Big)^{1-\frac{1}{\omega}}.
\end{equation*}
Note that if $\zeta\geq2k$, then
\begin{equation*}
\frac{2}{\sqrt{\zeta^{2}-k^{2}}}\bigg[\frac{\pi}{2}+\arctan\bigg(\frac{k}{\sqrt{\zeta^{2}-k^{2}}}\bigg)\bigg] >
\frac{\sqrt{3}}{\sqrt{\zeta^{2}-k^{2}}} \geq
\frac{1}{\sqrt{\zeta^{2}-k^{2}}}\hspace{1pt}\sqrt{\frac{\zeta+k}{\zeta-k}} =
\frac{1}{\zeta-k}\hspace{1pt}.
\end{equation*}
Moreover, if $k<\zeta<2k$, then
\begin{equation*}
\frac{2}{\sqrt{\zeta^{2}-k^{2}}}\bigg[\frac{\pi}{2}+\arctan\bigg(\frac{k}{\sqrt{\zeta^{2}-k^{2}}}\bigg)\bigg] >
\frac{2}{\sqrt{\zeta^{2}-k^{2}}} \geq
\frac{4}{\sqrt{\zeta^{2}-k^{2}}}\hspace{1pt}\sqrt{\frac{k^{2}}{\zeta^{2}}\bigg(1-\hspace{1pt}\frac{k^{2}}{\zeta^{2}}\bigg)} =
\frac{4k}{\zeta^{2}}\hspace{1pt}.
\end{equation*}
Therefore, Propositions \ref{Prop3.4.3} and \ref{Prop3.4.4} (with $\alpha=1$) ensure the boundedness of moments of order $\omega$ of the discounted process and its approximation. Furthermore, the convergence in probability of the discounted process is a simple consequence of Proposition \ref{Prop3.4.2}, by taking the domestic short rate to be zero. Finally, taking $\epsilon$ sufficiently small leads to the conclusion. %\hfill
\qquad
\end{proof}

%%%%%%%%%%%%%%%%%%%%%%%%%%%%%%%%%%%%%%%%%%%%%%%%%%%%%%%%%%%%%%%%%%%%%%%%%%%%%%
\subsection{Option valuation}\label{subsec:option}

We now examine the convergence of Monte Carlo estimators for computing FX option prices when the dynamics of the exchange rate are governed by the Heston--2CIR\scalebox{.9}{\raisebox{.5pt}{++}} SLV model and assumptions {\rm ($\mathcal{A}$1)} to {\rm ($\mathcal{A}$3)} are satisfied. As an aside, note that we discussed in Section \ref{sec:setup} how other derivative pricing models, including popular models in equity markets, can be formulated as special cases. First, we consider European options.

\begin{theorem}\label{Thm4.4.1} Let $P = \E\Big[e^{-\int_{0}^{T}{r^{d}_{t} dt}}\big(K-S_{T}\big)^{+}\Big]$ be the arbitrage-free price of a European put option and $\hspace{1.5pt}\bar{P} = \E\Big[e^{-\int_{0}^{T}{\bar{r}^{d}_{t} dt}}\big(K-\bar{S}_{T}\big)^{+}\Big]$ its approximation. If $2k_{f}\theta_{f}>\xi_{f}^{2}$, then
\begin{equation}\label{eq4.1}
\plim_{\delta t \to 0} \left|P - \bar{P}\right| = 0.
\end{equation}
\end{theorem}
\begin{proof}
A simple string of inequalities gives the following upper bound:
\begin{align}\label{eq4.1'}
\left|P - \bar{P}\right| &\leq \E\Big[\big|\big\{e^{-\int_{0}^{T}{r^{d}_{t} dt}}-e^{-\int_{0}^{T}{\bar{r}^{d}_{t} dt}}\big\}\big(K-S_{T}\big)^{+} + e^{-\int_{0}^{T}{\bar{r}^{d}_{t} dt}}\big\{\hspace{-1pt}\big(K-S_{T}\big)^{+}-\big(K-\bar{S}_{T}\big)^{+}\big\}\big|\Big] \nonumber \\[2pt]
&\leq Ke^{h_{max}T}\E\Big[\big|e^{-\int_{0}^{T}{g^{d}_{t} dt}}-e^{-\int_{0}^{T}{\bar{g}^{d}_{t} dt}}\big|\Big] + e^{h_{max}T}\E\left[\big|\big(K-S_{T}\big)^{+}-\big(K-\bar{S}_{T}\big)^{+}\big|\right].
\end{align}
However, for any non-negative numbers $x$ and $y$, $\left|e^{-x}-e^{-y}\right| \leq |x-y|$, and so we can use Fubini's theorem to obtain an upper bound for the first expectation,
\begin{align}\label{eq4.2}
\sup_{t\in[0,T]}\E\Big[\big|e^{-\int_{t}^{T}{g^{d}_{u} du}}-e^{-\int_{t}^{T}{\bar{g}^{d}_{u} du}}\big|\Big] &\leq \sup_{t\in[0,T]}\int_{t}^{T}{\E\left[|g^{d}_{u}-\bar{g}^{d}_{u}|\right]}du %\nonumber \\[1pt] &
\hspace{1pt}\leq\hspace{1pt} T\hspace{-1pt}\sup_{t \in [0,T]}\E\left[|g^{d}_{t}-\bar{g}^{d}_{t}|\right].
\end{align}
The right-hand side tends to zero by Proposition \ref{Prop3.3.5}. Define the two events $A = \big\{S_{T}<K\big\}$ and $\bar{A} = \big\{\bar{S}_{T}<K\big\}$, and denote the last expectation in \eqref{eq4.1'} by $\mathbb{J}$. Then
\begin{align*}
\mathbb{J} = \E\Big[\big|\big(K-S_{T}\big)^{+}-\big(K-\bar{S}_{T}\big)^{+}\big|\big(\Ind_{A\cap \bar{A}}+\Ind_{A\cap \bar{A}^{c}}+\Ind_{A^{c}\cap \bar{A}}+\Ind_{A^{c}\cap \bar{A}^{c}}\big)\Big].
\end{align*}
Therefore,
\begin{align}\label{eq4.3}
\mathbb{J} &\leq \E\Big[\big|S_{T} - \bar{S}_{T}\big|\Ind_{A\cap \bar{A}}\Big] + \E\Big[\big(K-S_{T}\big)\Ind_{A\cap \bar{A}^{c}}\Big] + \E\Big[\big(K-\bar{S}_{T}\big)\Ind_{A^{c}\cap \bar{A}}\Big] \nonumber \\[3pt]
&\leq \E\Big[\big|S_{T} - \bar{S}_{T}\big|\Ind_{A\cap \bar{A}}\Big] + K\Prob\left(A\cap \bar{A}^{c}\right) + K\Prob\left(A^{c}\cap \bar{A}\right).
\end{align}
Let $\delta$ be an arbitrary positive number, then we have the following inclusion of events:
\begin{align*}
A\cap \bar{A}^{c} &= \Big(\big\{S_{T}\leq K-\delta\big\} \cup \big\{K-\delta < S_{T} < K\big\}\Big) \cap \big\{\bar{S}_{T} \geq K\big\} \\[3pt]
&\subseteq \Big(\big\{S_{T}\leq K-\delta\big\} \cap \big\{\bar{S}_{T} \geq K\big\}\Big) \cup \big\{K-\delta < S_{T} < K\big\} \\[3pt]
&\subseteq \Big\{\big|S_{T}-\bar{S}_{T}\big| \geq \delta\Big\} \cup \big\{K-\delta < S_{T} < K\big\}.
\end{align*}
In terms of probabilities of events, we have
\begin{equation}\label{eq4.4}
\Prob\left(A\cap \bar{A}^{c}\right) \leq \Prob\Big(\big|S_{T}-\bar{S}_{T}\big| \geq \delta\Big) + \Prob\big(K-\delta < S_{T} < K\big),\hspace{1em} \forall \delta > 0.
\end{equation}
We can bound the second probability from above in a similar fashion,
\begin{align}\label{eq4.5}
A^{c}\cap \bar{A} &\subseteq \Big\{\big|S_{T}-\bar{S}_{T}\big| \geq \delta\Big\} \cup \big\{K \leq S_{T} < K+\delta\big\} \nonumber \\[4pt]
\Rightarrow \Prob\left(A^{c}\cap \bar{A}\right) &\leq \Prob\Big(\big|S_{T}-\bar{S}_{T}\big| \geq \delta\Big) + \Prob\big(K \leq S_{T} < K+\delta\big),\hspace{1em} \forall \delta > 0.
\end{align}
For a suitable choice of $\delta$, the last terms on the right-hand side of \eqref{eq4.4} and \eqref{eq4.5} can be made arbitrarily small, whereas the first terms tend to zero by Proposition \ref{Prop3.4.2}. Therefore, the two probabilities in \eqref{eq4.3} converge to zero as $\delta t \to 0$. Finally, fix $\epsilon > 0$ and let $B = \big\{|\hspace{1pt}S_{T}-\bar{S}_{T}| > \epsilon\big\}$. We can bound the expectation on the right-hand side of \eqref{eq4.3} as follows:
\begin{align}\label{eq4.6}
\E\Big[\big|S_{T}-\bar{S}_{T}\big|\Ind_{A\cap \bar{A}}\Big] &\leq \E\Big[\big|S_{T}-\bar{S}_{T}\big|\Ind_{A\cap \bar{A}}\Ind_{B^{c}}\Big] + \E\Big[\big|S_{T}-\bar{S}_{T}\big|\Ind_{A\cap \bar{A}}\Ind_{B}\Big] \nonumber \\[4pt]
&\leq K\Prob\Big(\big|S_{T}-\bar{S}_{T}\big| > \epsilon\Big) + \epsilon.
\end{align}
Taking the limit as $\delta t \to 0$, employing Proposition \ref{Prop3.4.2} and making use of the fact that $\epsilon$ can be made arbitrarily small leads to the conclusion. %\hfill
\qquad
\end{proof}

\begin{theorem}\label{Thm4.4.2} Let $C = \E\Big[e^{-\int_{0}^{T}{r^{d}_{t} dt}}\big(S_{T}-K\big)^{+}\Big]$ be the arbitrage-free price of a European call and $\hspace{1.5pt}\bar{C} = \E\Big[e^{-\int_{0}^{T}{\bar{r}^{d}_{t} dt}}\big(\bar{S}_{T}-K\big)^{+}\Big]$ its approximation. If $2k_{f}\theta_{f}>\xi_{f}^{2}$ and $T<T^{*}$, with $T^{*}$ from \eqref{eq3.43}, then
\begin{equation}\label{eq4.7}
\plim_{\delta t \to 0} \left|C - \bar{C}\right| = 0.
\end{equation}
\end{theorem}
\begin{proof}
A simple string of inequalities gives the following upper bound:
\begin{align}\label{eq4.8}
\left|C - \bar{C}\right| &\leq \E\Big[\big|\big(R_{T}-Ke^{-\int_{0}^{T}{r^{d}_{t}dt}}\big)^{+} - \big(\hspace{.5pt}\bar{R}_{T}-Ke^{-\int_{0}^{T}{\bar{r}^{d}_{t}dt}}\big)^{+}\big|\Big] \nonumber \\[2pt]
&\leq Ke^{h_{max}T}\E\Big[\big|e^{-\int_{0}^{T}{g^{d}_{t} dt}}-e^{-\int_{0}^{T}{\bar{g}^{d}_{t} dt}}\big|\Big] + \E\Big[\big|R_{T} - \bar{R}_{T}\big|\Big].
\end{align}
The two expectations on the right-hand side tend to zero as $\delta t \to 0$ from \eqref{eq4.2} and Theorem \ref{Thm3.4.6}, respectively. %\hfill
\qquad
\end{proof}

Asian options depend on the average exchange rate over a predetermined time period. Because the average is less volatile than the underlying rate, Asian options are usually less expensive than their European counterparts and are commonly used in currency and commodity markets, for instance, to reduce the foreign currency exposure of a corporation expecting payments in foreign currency. For any $0\leq s\leq t\leq T$, define the discount factors:
\begin{equation}\label{eq4.8'}
D_{s,\hspace{1pt}t}=e^{-\int_{s}^{t}{r^{d}_{u}du}} \hspace{3pt}\text{ and }\hspace{3pt} \bar{D}_{s,\hspace{1pt}t}=e^{-\int_{s}^{t}{\bar{r}^{d}_{u}du}}.
\end{equation}

\begin{theorem}\label{Thm4.4.3} Consider a fixed strike Asian option with arbitrage-free price
\begin{align*}
U &= \E\Big[e^{-\int_{0}^{T}{r^{d}_{t} dt}}\big[\psi(A(0,T)-K)\big]^{+}\Big], \\[2pt]
\hspace{1.5pt}\bar{U} &= \E\Big[e^{-\int_{0}^{T}{\bar{r}^{d}_{t} dt}}\big[\psi(\bar{A}(0,T)-K)\big]^{+}\Big].
\end{align*}
 If $2k_{f}\theta_{f}>\xi_{f}^{2}$ and $T<T^{*}$, with $T^{*}$ from \eqref{eq3.43}, then
\begin{equation}\label{eq4.9}
\plim_{\delta t \to 0} \left|U - \bar{U}\right| = 0.
\end{equation}
Here, $A(0,T)$ represents the arithmetic average and $\psi=\pm1$ depending on the payoff (call or put). For continuous monitoring, $A(0,T) = \frac{1}{T}\hspace{-.1em}\int_{0}^{T}{\hspace{-.1em}S_{t}\hspace{.1em}dt}$ and $\bar{A}(0,T) = \frac{1}{T}\hspace{-.1em}\int_{0}^{T}{\hspace{-.1em}\bar{S}_{t}\hspace{1pt}dt}$.
\end{theorem}
\begin{proof}
The absolute difference can be bounded from above by
\begin{align*}
\left|U - \bar{U}\right| \leq \E\Big[\big|\big[\psi(D_{0,T}A(0,T)-KD_{0,T})\big]^{+} - \big[\psi(\bar{D}_{0,T}\bar{A}(0,T)-K\bar{D}_{0,T})\big]^{+}\big|\Big].
\end{align*}
Therefore, we end up with the following upper bound:
\begin{align}\label{eq4.10}
\left|U - \bar{U}\right| \leq Ke^{h_{max}T}\E\Big[\big|e^{-\int_{0}^{T}{g^{d}_{t} dt}}-e^{-\int_{0}^{T}{\bar{g}^{d}_{t} dt}}\big|\Big] + \E\Big[\big|D_{0,T}A(0,T) - \bar{D}_{0,T}\bar{A}(0,T)\big|\Big].
\end{align}
We deduced the convergence of the first expectation in \eqref{eq4.2}. Using Fubini's theorem,
\begin{align*}
\E\Big[\big|D_{0,T}A(0,T) \hspace{-.5pt}-\hspace{-.5pt} \bar{D}_{0,T}\bar{A}(0,T)\big|\Big] &\leq \frac{1}{T}\E\bigg[\int_{0}^{T}{\big|D_{0,T}S_{t} - \bar{D}_{0,T}\bar{S}_{t}\big|\hspace{1.5pt}dt}\bigg] \\[2pt]
& \leq \sup_{t\in[0,T]}\E\Big[\big|D_{t,T}R_{t} - \bar{D}_{t,T}\bar{R}_{t}\big|\Big].
\end{align*}
The triangle inequality leads to the following upper bound:
\begin{align}\label{eq4.12}
\sup_{t\in[0,T]}\E\Big[\big|D_{t,T}R_{t} - \bar{D}_{t,T}\bar{R}_{t}\big|\Big] &\leq e^{h_{max}T}\sup_{t \in [0,T]}\E\Big[R_{t}\big|e^{-\int_{t}^{T}{g^{d}_{u} du}}-e^{-\int_{t}^{T}{\bar{g}^{d}_{u} du}}\big|\Big] \nonumber\\[1pt] 
&+ e^{h_{max}T}\sup_{t \in [0,T]}\E\Big[\big|R_{t}-\bar{R}_{t}\big|\Big].
\end{align}
Since both $g^{d}$ and $\bar{g}^{d}$ are non-negative processes, for any $\gamma$ greater than one we have
\begin{equation*}
\big|e^{-\int_{t}^{T}{g^{d}_{u} du}}-e^{-\int_{t}^{T}{\bar{g}^{d}_{u} du}}\big|^{\gamma} \leq \big|e^{-\int_{t}^{T}{g^{d}_{u} du}}-e^{-\int_{t}^{T}{\bar{g}^{d}_{u} du}}\big|\hspace{1pt},\hspace{1em} \forall\hspace{1pt} t\in[0,T].
\end{equation*}
Applying H\"older's inequality to the first expectation on the right-hand side of \eqref{eq4.12} with the pair $(\omega,\gamma)$, where $1<\omega<\omega_{1}$ and $\gamma = \omega/(\omega\hspace{-1pt}-\hspace{-1pt}1)$, and using the last inequality, we find that
\begin{equation}\label{eq4.13}
\sup_{t \in [0,T]}\E\Big[R_{t}\big|e^{-\int_{t}^{T}{g^{d}_{u} du}}-e^{-\int_{t}^{T}{\bar{g}^{d}_{u} du}}\big|\Big] \leq \sup_{t \in [0,T]}\E\big[R_{t}^{\omega}\big]^{\frac{1}{\omega}} \sup_{t \in [0,T]}\E\Big[\big|e^{-\int_{t}^{T}{g^{d}_{u} du}}-e^{-\int_{t}^{T}{\bar{g}^{d}_{u} du}}\big|\Big]^{\frac{1}{\gamma}}.
\end{equation}
The convergence of the first term on the right-hand side of \eqref{eq4.12} is a consequence of \eqref{eq4.2} and Proposition \ref{Prop3.4.3} (with $\alpha=1$), whereas the convergence of the second term is due to Theorem \ref{Thm3.4.6}. In case of discrete monitoring or a floating strike, we follow the exact same steps. %\hfill
\qquad
\end{proof}

Barrier options continue to gain popularity in many over-the-counter markets, including the FX market. Their popularity can be explained by two key factors. First, barrier options are useful in limiting the risk exposure of an investor in the FX market. Second, they offer additional flexibility and can match an investor's view on the market for a lower price than a vanilla option.

\begin{theorem}\label{Thm4.4.4} Consider an up-and-out barrier call with arbitrage-free price
\begin{align*}
U &= \E\Big[e^{-\int_{0}^{T}{r^{d}_{t}dt}}\big(S_{T}-K\big)^{+}\Ind_{\left\{\sup_{t \in [0,T]}S_{t}\leq B\right\}}\Big], \\[2pt]
\bar{U} &= \E\Big[e^{-\int_{0}^{T}{\bar{r}^{d}_{t}dt}}\big(\bar{S}_{T}-K\big)^{+}\Ind_{\big\{\sup_{t \in [0,T]} \bar{S}_{t}\leq B\big\}}\Big],
\end{align*}
where $K$ is the strike price and $B$ is the barrier. If $2k_{f}\theta_{f}>\xi_{f}^{2}$, then
\begin{equation}\label{eq4.14}
\plim_{\delta t \to 0} \left|U - \bar{U}\right| = 0.
\end{equation}
\end{theorem}
\begin{proof}
Define the events $A = \big\{\sup_{t \in [0,T]}S_{t} \leq B\big\}$ and $\bar{A} = \big\{\sup_{t \in [0,T]} \bar{S}_{t} \leq B\big\}$, then
\begin{align*}%\label{eq4.15}
\left|U-\bar{U}\right| &\leq \E\Big[\big|\big(D_{0,T}-\bar{D}_{0,T}\big)\big(S_{T}-K\big)^{+}\Ind_{A} %\nonumber \\[2pt] &\hspace{1em} 
+\hspace{1pt} \bar{D}_{0,T}\big\{\big(S_{T}-K\big)^{+}\Ind_{A}-\big(\bar{S}_{T}-K\big)^{+}\Ind_{\bar{A}}\big\}\big|\Big] \nonumber \\[3pt]
&\hspace{-2em}\leq \big(B-K\big)^{+}e^{h_{max}T}\E\Big[\big|e^{-\int_{0}^{T}{g^{d}_{t} dt}}-e^{-\int_{0}^{T}{\bar{g}^{d}_{t} dt}}\big|\Big] + e^{h_{max}T}\E\left[\big|\big(S_{T}-K\big)^{+}\Ind_{A}-\big(\bar{S}_{T}-K\big)^{+}\Ind_{\bar{A}}\big|\right].
\end{align*}
The first term tends to zero by \eqref{eq4.2} and we can rewrite the second term as follows:
\begin{align}\label{eq4.16}
&\E\Big[\big|\big(S_{T}-K\big)^{+}\big(\Ind_{A\cap \bar{A}^{c}}+\Ind_{A\cap \bar{A}}\big)-\big(\bar{S}_{T}-K\big)^{+}\big(\Ind_{A\cap \bar{A}}+\Ind_{A^{c}\cap \bar{A}}\big)\big|\Big] \nonumber \\[3pt]
&\hspace{4em}\leq \E\left[\big(S_{T}-K\big)^{+}\Ind_{A\cap \bar{A}^{c}}\right] + \E\Big[\big(\bar{S}_{T}-K\big)^{+}\Ind_{A^{c}\cap \bar{A}}\Big] + \E\Big[\big|S_{T}-\bar{S}_{T}\big|\Ind_{A\cap \bar{A}}\Big] \nonumber \\[3pt]
&\hspace{4em}\leq \big(B-K\big)^{+}\Big\{\Prob\left(A\cap \bar{A}^{c}\right)+\Prob\left(A^{c}\cap \bar{A}\right)\Big\} + \E\Big[\big|S_{T}-\bar{S}_{T}\big|\Ind_{A\cap \bar{A}}\Big].
\end{align}
We can bound the last expectation from above just as in \eqref{eq4.6} to find
\begin{equation}\label{eq4.17}
\E\Big[\big|S_{T}-\bar{S}_{T}\big|\Ind_{A\cap \bar{A}}\Big] \leq B\Prob\Big(\big|S_{T}-\bar{S}_{T}\big| > \epsilon\Big) + \epsilon, \hspace{1em}\forall\hspace{1pt} \epsilon>0.
\end{equation}
Therefore, the expectation converges to zero with the time step by Proposition \ref{Prop3.4.2}. Fixing $\delta>0$ and following the argument of Theorem 6.2 in \cite{Higham:2005} leads to
\begin{align*}
A\cap \bar{A}^{c} \subseteq \bigg\{\sup_{t \in [0,T]}\big|S_{t}-\bar{S}_{t}\big| \geq \delta\bigg\} \cup \bigg\{B-\delta < \sup_{t \in [0,T]}S_{t} \leq B\bigg\}.
\end{align*}
In terms of probabilities of events, we have
\begin{equation}\label{eq4.18}
\Prob\left(A\cap \bar{A}^{c}\right) \leq \Prob\bigg(\sup_{t \in [0,T]}\big|S_{t}-\bar{S}_{t}\big| \geq \delta\bigg) + \Prob\bigg(B-\delta < \sup_{t \in [0,T]}S_{t} \leq B\bigg).
\end{equation}
We can bound the second probability in \eqref{eq4.16} from above in a similar fashion,
\begin{align}\label{eq4.19}
\Prob\left(A^{c}\cap \bar{A}\right) &\leq \Prob\bigg(\sup_{t \in [0,T]}\big|S_{t}-\bar{S}_{t}\big| \geq \delta\bigg) + \Prob\bigg(B < \sup_{t \in [0,T]}S_{t} < B + \delta\bigg).
\end{align}
The conclusion follows from Proposition \ref{Prop3.4.2} since $\delta$ can be arbitrarily small. %\hfill
\qquad
\end{proof}

\begin{theorem}\label{Thm4.4.5} Consider any type of barrier put option with arbitrage-free price
\begin{align*}
U &= \E\left[e^{-\int_{0}^{T}{r^{d}_{t} dt}}\big(K-S_{T}\big)^{+}\Ind_{A}\right], \\[3pt]
\bar{U} &= \E\Big[e^{-\int_{0}^{T}{\bar{r}^{d}_{t} dt}}\big(K-\bar{S}_{T}\big)^{+}\Ind_{\bar{A}}\Big],
\end{align*}
where the events $A$ and $\bar{A}$ depend on the type of barrier. If $2k_{f}\theta_{f}>\xi_{f}^{2}$, then
\begin{equation}\label{eq4.20}
\plim_{\delta t \to 0} \left|U - \bar{U}\right| = 0.
\end{equation}
For instance, a down-and-in barrier is associated with the set $A = \big\{\inf_{t \in [0,T]}S_{t} \leq B\big\}$.
\end{theorem}
\begin{proof}
An upper bound for the absolute difference can be obtained as follows:
\begin{align*}
\left|U\hspace{-1pt}-\bar{U}\right| &\leq\E\Big[\big|\big(D_{0,T}-\bar{D}_{0,T}\big)\big(K\hspace{-.5pt}-\hspace{-.5pt}S_{T}\big)^{+}\Ind_{A} +\hspace{.5pt} \hspace{1pt} \bar{D}_{0,T}\big\{\big(K\hspace{-.5pt}-\hspace{-.5pt}S_{T}\big)^{+}\Ind_{A}-\big(K\hspace{-.5pt}-\hspace{-.5pt}\bar{S}_{T}\big)^{+}\Ind_{\bar{A}}\hspace{-2pt}\big\}\big|\Big] \nonumber \\[4pt]
&\leq Ke^{h_{max}T}\E\Big[\big|e^{-\int_{0}^{T}{g^{d}_{t} dt}}-e^{-\int_{0}^{T}{\bar{g}^{d}_{t} dt}}\big|\Big] + e^{h_{max}T}\E\left[\big|\big(K-S_{T}\big)^{+}\Ind_{A}-\big(K-\bar{S}_{T}\big)^{+}\Ind_{\bar{A}}\big|\right].
\end{align*}
The first term tends to zero by \eqref{eq4.2} and we can bound the second term as in \eqref{eq4.16}:
\begin{align}\label{eq4.22}
\E\left[\big|\big(K-S_{T}\big)^{+}\Ind_{A}-\big(K-\bar{S}_{T}\big)^{+}\Ind_{\bar{A}}\big|\right] &\leq K\Big\{\Prob\left(A\cap \bar{A}^{c}\right)+\Prob\left(A^{c}\cap \bar{A}\right)\Big\} \nonumber \\[3pt]
&+ \E\left[\big|\big(K-S_{T}\big)^{+}-\big(K-\bar{S}_{T}\big)^{+}\big|\right].
\end{align}
The events $A$ and $\bar{A}$ differ with the barrier (down-and-in, down-and-out, up-and-in, up-and-out), however one can show in a similar way to \eqref{eq4.18} and \eqref{eq4.19} that
\begin{equation}\label{eq4.23}
\plim_{\delta t \to 0}\Prob\left(A\cap \bar{A}^{c}\right)=0 \; \text{ and } \; \plim_{\delta t \to 0}\Prob\left(A^{c}\cap \bar{A}\right)=0
\end{equation}
for any type of barrier. Finally, the convergence of the last term on the right-hand side of \eqref{eq4.22} was derived in Theorem \ref{Thm4.4.1}, which concludes the proof. %\hfill
\qquad
\end{proof}

\begin{theorem}\label{Thm4.4.6} Consider a down-and-in/out or up-and-in barrier call option with arbitrage-free price
\begin{align*}
U &= \E\left[e^{-\int_{0}^{T}{r^{d}_{t} dt}}\big(S_{T}-K\big)^{+}\Ind_{A}\right], \\[3pt]
\bar{U} &= \E\Big[e^{-\int_{0}^{T}{\bar{r}^{d}_{t} dt}}\big(\bar{S}_{T}-K\big)^{+}\Ind_{\bar{A}}\Big],
\end{align*}
where the events $A$ and $\bar{A}$ depend on the type of barrier. If $2k_{f}\theta_{f}>\xi_{f}^{2}$ and $T<T^{*}$, with $T^{*}$ from \eqref{eq3.43}, then
\begin{equation}\label{eq4.24}
\plim_{\delta t \to 0} \left|U - \bar{U}\right| = 0.
\end{equation}
\end{theorem}
\begin{proof}
An upper bound for the absolute difference can be obtained as follows:
\begin{align*}
\left|U-\bar{U}\right| &\leq \E\Big[\big|\big(R_{T}-KD_{0,T}\big)^{+} - \big(\bar{R}_{T}-K\bar{D}_{0,T}\big)^{+}\big|\Ind_{A\cap\bar{A}}\Big] \\[3pt]
&+ \E\Big[\big(R_{T}-KD_{0,T}\big)^{+}\Ind_{A\cap\bar{A}^{c}}\Big] + \E\Big[\big(\bar{R}_{T}-K\bar{D}_{0,T}\big)^{+}\Ind_{A^{c}\cap\bar{A}}\Big].
\end{align*}
Therefore, we end up with
\begin{align*}
\left|U-\bar{U}\right| \leq \E\!\Big[\big|R_{T}-\bar{R}_{T}\big|\Big] + K\hspace{-1pt}e^{h_{max}T}\E\!\Big[\big|e^{-\int_{0}^{T}{g^{d}_{t} dt}}-e^{-\int_{0}^{T}{\bar{g}^{d}_{t} dt}}\big|\Big] + \E\!\Big[R_{T}\Ind_{A\cap\bar{A}^{c}}\Big] + \E\!\Big[\bar{R}_{T}\Ind_{A^{c}\cap\bar{A}}\Big].
\end{align*}
The convergence of the first two terms on the right-hand side is a consequence of Theorem \ref{Thm3.4.6} and \eqref{eq4.2}, respectively. Applying H\"older's inequality with the pair $(\omega,\gamma)$ to the last two terms, where $1<\omega<\min\{\omega_{1},\omega_{2}\}$, we find that:
\begin{align*}
\E\Big[R_{T}\Ind_{A\cap\bar{A}^{c}}\Big] \leq \E\Big[\big(R_{T}\big)^{\omega}\Big]^{\frac{1}{\omega}}\Prob\left(A\cap \bar{A}^{c}\right)^{\frac{1}{\gamma}}
\end{align*}
and
\begin{align*}
\E\Big[\bar{R}_{T}\Ind_{A^{c}\cap\bar{A}}\Big] \leq \E\Big[\big(\hspace{.5pt}\bar{R}_{T}\big)^{\omega}\Big]^{\frac{1}{\omega}}\Prob\left(A^{c}\cap \bar{A}\right)^{\frac{1}{\gamma}}.
\end{align*}
Using Propositions \ref{Prop3.4.3} and \ref{Prop3.4.4} (with $\alpha=1$) and the limits in \eqref{eq4.23} concludes the proof. %\hfill
\qquad
\end{proof}

Among the most developed exotic derivatives in the FX market are the double barrier options.

\begin{theorem}\label{Thm4.4.7} Consider a double knock-out call option with arbitrage-free price
\begin{align*}
U &= \E\left[e^{-\int_{0}^{T}{r^{d}_{t} dt}}\big(S_{T}-K\big)^{+}\Ind_{\left\{\inf_{t \in [0,T]} S_{t} \geq L,\;\sup_{t \in [0,T]} S_{t} \leq B\right\}}\right], \\[3pt]
\bar{U} &= \E\Big[e^{-\int_{0}^{T}{\bar{r}^{d}_{t} dt}}\big(\bar{S}_{T}-K\big)^{+}\Ind_{\big\{\inf_{t \in [0,T]} \bar{S}_{t} \geq L,\;\sup_{t \in [0,T]} \bar{S}_{t} \leq B\big\}}\Big],
\end{align*}
where $K$ is the strike and $L$, $B$ are the lower and upper barriers, respectively. If $2k_{f}\theta_{f}>\xi_{f}^{2}$, then
\begin{equation}\label{eq4.26}
\plim_{\delta t \to 0} \left|U - \bar{U}\right| = 0.
\end{equation}
\end{theorem}
\begin{proof}
First, note that we have the following inclusion of events:
\begin{align*}
&\big\{\hspace{-.2em}\inf_{t \in [0,T]} S_{t} \geq L,\;\sup_{t \in [0,T]} S_{t} \leq B\big\} \cap \big\{\hspace{-.2em}\inf_{t \in [0,T]} \bar{S}_{t} \geq L,\;\sup_{t \in [0,T]} \bar{S}_{t} \leq B\big\}^{c} \\[2pt]
&\hspace{3em} \subseteq \Big(\big\{\hspace{-.3em}\sup_{t \in [0,T]} S_{t} \leq B\big\}\cap\big\{\hspace{-.3em}\sup_{t \in [0,T]} \bar{S}_{t} > B\big\}\Big) \cup \Big(\big\{\hspace{-.2em}\inf_{t \in [0,T]} S_{t} \geq L\big\}\cap\big\{\hspace{-.2em}\inf_{t \in [0,T]} \bar{S}_{t} < L\big\}\Big).
\end{align*}
The rest of the argument follows closely that of Theorem \ref{Thm4.4.4} and is thus omitted. %\hfill
\qquad
\end{proof}

\section{Numerical results}\label{sec:numerics}

In this section, we consider the pricing problem for exotic products, in particular, for structured notes embedding barrier features. Popular FX exotic products include the power reverse dual currency note (PRDC) \cite{Clark:2011}, especially long-term (e.g., 30 years), and the forward accumulator \cite{Wystup:2007}. Adding stochastic rates improves the pricing of both contracts for mid- to long-term expiries. Moreover, while the plain PRDC can be seen as a strip of vanilla options, adding stochastic-local volatility dynamics improves the pricing of some variations of the contract, like the trigger PRDC (a PRDC with a knock-out feature).

Henceforth, we consider an autocallable barrier dual currency note (ABDC) that is indexed on the EURUSD currency pair, with USD as the domestic currency and EUR as the foreign currency, with quarterly coupon payments. While barrier dual currency notes are typically short-term investments, we examine instead a mid-term alternative with an embedded autocallable note structure. For a nominal $N$ (in USD), a rate (strike) $K$ and an expiry $T$, the life cycle of the contract is described below.
\begin{enumerate}[(1)]
\item{A monthly fixing schedule is defined for the exchange rate $S$ such that:
\begin{itemize}
\item If $S$ crosses the up barrier $B_{\text{UO}}$, the contract is redeemed early and a coupon $\mathcal{C}_{\text{ER}}$ is paid. The early redemption time is denoted by $\tau_{\text{ER}}$ and can be infinite in case of no knock-out event.
\item If $S$ crosses the down barrier $B_{\text{DI}}$, a short put contract is activated at expiry. The knock-in time is denoted by $\tau_{\text{KI}}$.
\end{itemize}}
\item{If $S$ is above $B_{\text{DI}}$, a coupon $\mathcal{C}$ (in USD, in \% of the nominal) is paid quarterly at coupon dates $\left\{t_{i}\hspace{1pt},\hspace{2pt} i=1,\hdots,M\right\}$, for a total of $M$ periods, such that the accumulated coupon (in USD) at expiry is
\begin{equation}\label{eq6.1}
N\sum_{i=1}^{M}{\frac{D_{t_{i}}^{d}}{D_{T}^{d}}\hspace{1pt}\mathcal{C}\Ind_{t_{i}<\tau_{\text{ER}}}\Ind_{S_{t_{i}}>B_{\text{DI}}}}
+ N\frac{D_{\tau_{\text{ER}}}^{d}}{D_{T}^{d}}\hspace{1pt}\mathcal{C}_{\text{ER}}\Ind_{\tau_{\text{ER}}\leq T}.
\end{equation}}
\item{At expiry, if $S$ is below $K$, the down barrier has been activated and early redemption has not occurred, then the nominal is converted to EUR at the rate $K$.}
\end{enumerate}
This product would suit an investor who wants to outperform the money market account by taking the risk of conversion of the nominal value to the foreign currency. The autocall feature lowers the price of the product due to early termination in case the exchange rate crosses the up-and-out barrier, whereas the down-and-in barrier feature provides protection against the conversion risk at the expense of increasing the net present value (NPV). Hence, an investor would find this product attractive under the current volatile market conditions if they expected a less volatile market in the future, which would imply a lower NPV due to the increased chance of early redemption and conversion. Over a longer time horizon, a very stable market would be the best outcome for the investor, since they would get all the coupons without converting the nominal at maturity. If conversion occurs at expiry, the investor receives $N/K$ in EUR in exchange for the nominal $N$ in USD. If the investor converts this amount to USD, the loss becomes $N\big(S_{T}/K-1\big)$. Hence, the profit-and-loss (PnL) of the product is
\begin{equation}\label{eq6.2}
N\sum_{i=1}^{M}{\frac{D_{t_{i}}^{d}}{D_{T}^{d}}\hspace{1pt}\mathcal{C}\Ind_{t_{i}<\tau_{\text{ER}}}\Ind_{S_{t_{i}}>B_{\text{DI}}}}
+ N\frac{D_{\tau_{\text{ER}}}^{d}}{D_{T}^{d}}\hspace{1pt}\mathcal{C}_{\text{ER}}\Ind_{\tau_{\text{ER}}\leq T}
-\hspace{1pt} \frac{N}{K}(K-S_{T})^{+}\Ind_{\tau_{\text{KI}}\leq T<\tau_{\text{ER}}}.
\end{equation}
Hence, this product is a yield enhancement contract with no capital guarantee. The accrued coupon is sometimes converted at expiry with the nominal at the rate $K$. Finally, the NPV of the contract (in \% of the nominal) is
\begin{equation}\label{eq6.3}
\text{NPV} = \E\!\Bigg[\sum_{i=1}^{M}{D_{t_{i}}^{d}\mathcal{C}\Ind_{t_{i}<\tau_{\text{ER}}}\Ind_{S_{t_{i}}>B_{\text{DI}}}} + D_{\tau_{\text{ER}}}^{d}\mathcal{C}_{\text{ER}}\Ind_{\tau_{\text{ER}}\leq T} - \frac{D_{T}^{d}}{K}\left(K-S_{T}\right)^{+}\Ind_{\tau_{\text{KI}}\leq T<\tau_{\text{ER}}}\Bigg],
\end{equation}
where the expectation is taken under the risk-neutral measure. We assume the dynamics of $S$ under this measure as specified in \eqref{eq2.1}.

Next, consider the contract parameters from Table \ref{table6.1} as well as the calibrated model parameters and leverage function -- to the EURUSD market data from March 18, 2016 -- from \cite{Mariapragassam:2016}. Suppose that barriers are monitored monthly, at the fixing dates, and coupons are paid quarterly.
\begin{table}[htb]
\begin{center}
\caption{The contract parameters of an autocallable barrier dual currency note, where the nominal $N$ is in USD, the expiry $T$ is in years, the strike $K$ and the barriers $B_{\text{UO}}$ and $B_{\text{DI}}$ are in \% of $S_{0}$, and the coupons $\mathcal{C}$ and $\mathcal{C}_{\text{ER}}$ are in USD, in \% of $N$.}\label{table6.1}
\begin{tabularx}{\textwidth}{@{}YYYYYYYY@{}}
  \addlinespace[-5pt]
	\toprule[.1em]
  $N$ & $T$ & $S_{0}$ & $K$ & $B_{\text{UO}}$ & $B_{\text{DI}}$ & $\mathcal{C}$ & $\mathcal{C}_{\text{ER}}$ \\
  \midrule
  $100\hspace{1pt}000$ & $5$Y & $1.1271$ & $105$\% & $100$\% & $95$\% & $2.5$\% & $1.5$\% \\
	\bottomrule[.1em]
	\addlinespace[3pt]
\end{tabularx}
\end{center}
\end{table}

We employ the Monte Carlo simulation scheme defined in Section~\ref{sec:analysis} with $5\scalebox{0.85}{$\times$}10^{7}$ sample paths and $376$ time steps per year to price this contract, and the numerical results are displayed in Table \ref{table6.2}. Note that we computed the percentage change in premium (NPV) with respect to the reference model, i.e., the Heston--2CIR\scalebox{.9}{\raisebox{.5pt}{++}} SLV model.
\begin{table}[htb]\addtolength{\tabcolsep}{-4pt}
\begin{center}
\caption{The NPV of the contract in \% of $N$, the $95$\% Monte Carlo confidence interval, the early redemption and knock-in probabilities and the percentage change in NPV, for the autocallable barrier dual currency note specified in Table \ref{table6.1} under the 4-factor hybrid SLV model \eqref{eq2.1}, the 2-factor SLV model with deterministic rates and the LV model \eqref{eq2.2.c}.}\label{table6.2}
\begin{tabularx}{\textwidth}{@{}YYYYYY@{}}
\addlinespace[-5pt]
  \toprule[.1em]
	\textbf{Model} & \textbf{NPV} & \textbf{$95$\% CI} & \textbf{ER prob.} & \textbf{KI prob.} & \textbf{Change} \\
  \midrule
  \textbf{Hybrid SLV} & $1.7247$ & $(1.7228,1.7265)$ & $95.1$\% & $19.5$\% & -- \\[1pt]%\hspace{-3pt}(..)\hspace{-3pt}
  \textbf{SLV} & $1.7456$ & $(1.7438, 1.7474)$ & $95.1$\% & $19.5$\% & $+1.21$\% \\[1pt]
  \textbf{LV} & $1.2079$ & $(1.2062,1.2097)$ & $95.9$\% & $19.5$\% & $-29.96$\% \\[1pt]
  \bottomrule[.1em]
  \addlinespace[3pt]
\end{tabularx}
\end{center}
\end{table}

We infer from Table \ref{table6.2} that the price of the contract under the 2-factor SLV model is $1.21$\% higher than under the 4-factor SLV model, whereas the early redemption and the knock-in probabilities are the same under the two models. This suggests that the stochastic rates have a significant impact on the price of the contract even for a 5-year expiry. We also infer from Table \ref{table6.2} that the contract is highly underpriced and the early redemption probability is overestimated under the LV model \eqref{eq2.2.c}. Furthermore, we calibrated in \cite{Mariapragassam:2016} the 4-factor SLV and SV models to EURUSD market data from March 18, 2016, and observed an almost perfect fit to vanilla options with the former as opposed to a poor fit with the latter. In conclusion, we notice two things. First, pricing exotic products with barrier features under stochastic-local volatility dynamics is of paramount importance. It is a well-known fact that pure SV models underestimate the knock-out probability whereas pure LV models overestimate it, and the true price of the contract is believed to lie in between pure SV and pure LV model prices. In our case, a small difference in the probability of no knock-out, which is leveraged by the number of coupons detached during the life cycle of the contract, can lead to a big difference in the NPV, and this explains the need to use SLV models to improve the pricing performance. Second, adding stochastic rates to an SLV model is only relevant when pricing mid- to long-term structured notes, where the effect of the stochastic rates is accentuated. The impact of stochastic rates on the contract price is expected to increase in the presence of non-zero correlations between the spot FX rate and the short rates.

A small difference in the probability of no knock-out gives rise to a big difference in the NPV since it is leveraged by the number of coupons detached during the lifecycle of the contract.

We could extend the model \eqref{eq2.1} to multi-factor short rates when pricing exotic products where the rates appear explicitly in the payoff, for example, a spread option with the payoff
\begin{equation}\label{eq6.4}
\left[\left(\frac{S_{T}-S_{0}}{S_{0}}\right)-L_{T}-K\right]^{+},
\end{equation}
where $L_{T}$ the Libor rate at the fixing date $T$. In this case, the calibration algorithm described in Section~\ref{sec:setup} can still be applied, with a higher computational cost due to the more complex simulation scheme.

We conclude this section with an empirical convergence analysis of our Monte Carlo simulation scheme. Since the product can be decomposed into a linear combination of first-order exotics, the theoretical convergence (without a rate) of Monte Carlo estimators follows automatically from the analysis in Section~\ref{sec:analysis}. The data in Table \ref{table6.3} suggest a first-order convergence of the time-discretization error. As an aside, note that in the case of continuously monitored barriers, we could use Brownian bridge techniques to recover the first-order convergence.
\begin{table}[htb]
\begin{center}
\caption{The Monte Carlo estimates of the NPV of the contract specified in Table \ref{table6.1} for different numbers of time steps (per year) and $5\!\times\!10^{7}$ sample paths (for a standard deviation of $9.29\!\times\!10^{-4}$), the difference from the previous NPV estimate and the empirical convergence order.}\label{table6.3}
\begin{tabularx}{\textwidth}{@{}YYYY@{}}
	\addlinespace[-5pt]
	\toprule[.1em]
	\textbf{Time steps} & \textbf{NPV} & \textbf{Difference} & \textbf{Order} \\
	\midrule
	$12$ & $1.9081$ & -- & -- \\
	$24$ & $1.8110$ & $0.0971$ & -- \\
	$48$ & $1.7623$ & $0.0487$ & $0.996$ \\
	$96$ & $1.7388$ & $0.0234$ & $1.057$ \\
	$\hspace{-5pt}192$ & $1.7285$ & $0.0103$ & $1.184$ \\
	\bottomrule[.1em]
	\addlinespace[3pt]
\end{tabularx}
\end{center}
\end{table}

The empirical findings of this section demonstrate the importance of stochastic-local volatility dynamics as well as stochastic short rate dynamics for the pricing of long-dated exotic FX products. Furthermore, we verified the convergence of our Monte Carlo simulation scheme for one such product.

%%%%%%%%%%%%%%%%%%%%%%%%%%%%%%%%%%%%%%%%%%%%%%%%%%%%%%%%%%%%%%%%%%%%%%%%%%%%%%
%%% Section 5 %%%%%%%%%%%%%%%%%%%%%%%%%%%%%%%%%%%%%%%%%%%%%%%%%%%%%%%%%%%%%%%%%%%%%%
%%%%%%%%%%%%%%%%%%%%%%%%%%%%%%%%%%%%%%%%%%%%%%%%%%%%%%%%%%%%%%%%%%%%%%%%%%%%%%
\section{Conclusions}\label{sec:conclusion}

Our aim was to establish the strong convergence of an Euler scheme for a hybrid stochastic-local volatility model. The only previous published work related to this problem that we are aware of is \cite{Higham:2005}, which proves the convergence of an Euler discretization with a reflection fix in the context of Heston's model and options with bounded payoffs. We established the strong $L^{1}$-convergence (without a rate) of the discounted exchange rate approximation, a result which can be generalized to the $L^{p}$ case relatively easily, for all $p\geq1$, albeit under a stronger condition on the maturity, and which is particularly useful when proving the convergence of Monte Carlo simulations for valuing options with unbounded payoffs.
%Our aim was to extend the convergence theory for the full truncation scheme \cite{Lord:2010} from the square-root process to a hybrid stochastic-local volatility model.

The analysis carried out in this paper can be extended to other financial derivatives, including digital options, forward-start options and also double-no-touch binary options, to name just a few. Furthermore, we may consider a multi-factor extension of the short rate model, in which case the convergence analysis applies with some slight modifications of the proofs.

However, several unsettled questions remain, like the exact strong convergence rate of the full truncation scheme for the CIR process, or the strong convergence rate of schemes for the type of SDEs studied in this paper. On top of these being interesting and practically relevant questions in their own right, a sufficiently high order enables the use of multi-level simulation, as in \cite{Giles:2009}, with substantial efficiency improvements for the estimation of expected financial payoffs.

%%%%%%%%%%%%%%%%%%%%%%%%%%%%%%%%%%%%%%%%%%%%%%%%%%%%%%%%%%%%%%%%%%%%%%%%%%%%%%
%%% Bibliography %%%%%%%%%%%%%%%%%%%%%%%%%%%%%%%%%%%%%%%%%%%%%%%%%%%%%%%%%%%%%%%%%%%%
%%%%%%%%%%%%%%%%%%%%%%%%%%%%%%%%%%%%%%%%%%%%%%%%%%%%%%%%%%%%%%%%%%%%%%%%%%%%%%
\bibliographystyle{siam}
\bibliography{references}

%%%%%%%%%%%%%%%%%%%%%%%%%%%%%%%%%%%%%%%%%%%%%%%%%%%%%%%%%%%%%%%%%%%%%%%%%%%%%%
%%% End of Document %%%%%%%%%%%%%%%%%%%%%%%%%%%%%%%%%%%%%%%%%%%%%%%%%%%%%%%%%%%%%%%%%%
%%%%%%%%%%%%%%%%%%%%%%%%%%%%%%%%%%%%%%%%%%%%%%%%%%%%%%%%%%%%%%%%%%%%%%%%%%%%%%
\end{document}